\documentclass[11pt]{article}

\usepackage[utf8]{inputenc}
\usepackage[english]{babel}
\usepackage{amsmath}
\usepackage{nicefrac}%The nor­mal way to type­set frac­tions is with a
\usepackage{amsthm}
\usepackage{amsfonts}
\usepackage{amssymb}
\usepackage{verbatim}
\usepackage{color}
\usepackage[arrow, matrix, curve]{xy}
\usepackage{bbm}%For mathbbmss and mathbbm
\usepackage{eqparbox}
\usepackage[numbers]{natbib}
\usepackage{stmaryrd}
\usepackage{amssymb}
\usepackage{mathrsfs}
\usepackage{pictexwd,dcpic}
\usepackage{mathabx}

\DeclareMathSymbol{\leq}{\mathrel}{symbols}{20}
   
\DeclareMathSymbol{\geq}{\mathrel}{symbols}{21}

\newtheoremstyle{WreschTheoremstyle} % Name
                        {1.5em}    % Space above
                        {2.5em}    % Space below
                        {}         % Body font
                        {}         % Indent amount
                        {\bfseries}% Theorem head font
                        {}        % Punctuation after theorem head
                        {\newline} % Space after theorem head
                        {\raisebox{0.6em}{\thmname{#1}\thmnumber{#2}\thmnote{ (#3)}}}% Theorem head spec (can be left empty, meaning 'normal')

\newcommand{\R}{\mathbb{R}}
\newcommand{\C}{\mathbb{C}}
\newcommand{\N}{\mathbb{N}}

\newcommand{\K}{\mathcal{K}}

\newcommand{\e}{\varepsilon}
\newcommand{\Lb}{\mathcal{L}}
\renewcommand{\1}{\mathbbm{1}}
\newcommand{\dm}{\mathrm{d}}
\newcommand{\esssup}{\mathrm{ess}\sup}

\newcommand{\bfrac}[2]{\genfrac{}{}{0pt}{}{#1}{#2}}

\newtheorem{theorem}{theorem}[section]
\newtheorem{proposition}[theorem]{proposition}

\newtheorem{lemma}[theorem]{Lemma}
\newtheorem{remark}[theorem]{Remark}

\numberwithin{equation}{section}

\makeatletter
\newcommand{\customlabel}[1]{%
     \stepcounter{ref}%
   \protected@write
\@auxout{}{\string\newlabel{#1}{{\thesatz.\arabic{ref}}{\thepage}{\thesatz.\arabic{ref}}{#1}{}}}%
   \hypertarget{#1}{\thesatz.\arabic{ref}}%
}
\makeatother

\newenvironment{sciabstract}{\begin{quote}}{\end{quote}}

\newcounter{lastnote}

\title{Stochastic averaging principle for spatial Markov evolutions in the continuum}
\newcommand{\pdftitle} {Stochastic averaging principle for spatial Markov evolutions in the continuum}
\newcommand{\pdfauthor}{Martin Friesen}

\author{
Martin Friesen\footnote{Department of Mathematics, Wuppertal University, Germany, friesen@uni-wuppertal.de}\\
Yuri Kondratiev\footnote{Department of Mathematics, Bielefeld University, Germany, kondrat@math.uni-bielefeld.de}
}

\usepackage[plainpages=false,pdfpagelabels=true,bookmarks=true,pdfauthor={\pdfauthor},
pdftitle={\pdftitle}]{hyperref}

\makeatletter
\def\HyPsd@CatcodeWarning#1{}
\makeatother

\begin{document}


%\baselineskip24pt


\maketitle

\begin{sciabstract}\textbf{Abstract:}
We study a spatial birth-and-death process on the phase space of locally finite configurations $\Gamma^+ \times \Gamma^-$ over $\R^d$.
Dynamics is described by an non-equilibrium evolution of states obtained from the Fokker-Planck equation and 
associated with the Markov operator $L^+(\gamma^-) + \frac{1}{\e}L^-$, $\e > 0$.
Here $L^-$ describes the environment process on $\Gamma^-$ and $L^+(\gamma^-)$ describes the system process on $\Gamma^+$,
where $\gamma^-$ indicates that the corresponding birth-and-death rates depend on another locally finite configuration $\gamma^- \in \Gamma^-$. 
We prove that, for a certain class of birth-and-death rates, the corresponding Fokker-Planck equation is well-posed,
i.e. there exists a unique evolution of states $\mu_t^{\e}$ on $\Gamma^+ \times \Gamma^-$.
Moreover, we give a sufficient condition such that the environment is ergodic with exponential rate.
Let $\mu_{\mathrm{inv}}$ be the invariant measure for the environment process on $\Gamma^-$.
In the main part of this work we establish the stochastic averaging principle, i.e. we prove that the marginal 
of $\mu_t^{\e}$ onto $\Gamma^+$ converges weakly to an evolution of states on $\Gamma^+$
associated with the averaged Markov birth-and-death operator $\overline{L} = \int_{\Gamma^-}L^+(\gamma^-)d \mu_{\mathrm{inv}}(\gamma^-)$.
\end{sciabstract}


\section{Introduction}
Many physical, ecological and biological phenomena can be modelled by spatial birth-and-death processes.
It is assumed that particles are located in a continuous space, say $\R^d$, are identical by properties, indistinguishable and,
randomly appear or disappear in the location space. Particular examples can be found in \cite{BCFKKO14}, \cite{FFHKKK15},\cite{SEW05}, see also the references therein. 

Models where it is assumed that the total number of particles remains finite at any moment of time
are fairly well-understood (see e.g. \cite{EW03}, \cite{FM04}, \cite{KOL06}).
In the case where the collection of all particles forms only a locally finite configuration in $\R^d$ the situation is much more subtle.
Existence of a Markov process is only known for the case where the death rate is constant and the birth rate is sublinear (see \cite{GKT06}, \cite{KS06}).
The main difficulty comes from the necessity to control the number of particles in any bounded volume.
However, in many cases it is possible to study the corresponding one-dimensional distributions in terms of solutions to the Fokker-Planck equation
under more broad assumptions on the birth-and-death rates (see e.g. \cite{FKK15} and the references given therein).
Such a construction is based on the use of correlation functions which are supposed to satisfy the Ruelle bound
and are obtained from a Markov analogue of the well-known BBGKY-hierarchy. 
A general description of this approach goes back to \cite{FKO09}, \cite{KKM08}.

This work is devoted to the study of the Fokker-Planck equation for a general class of models consisting of a birth-and-death process (system) with
rates depending on another birth-and-death process (environment) in the space of locally finite configurations over $\R^d$.
We suppose that the environment has, compared to the system, significantly large birth-and-death rates, i.e.
we consider the scaling regime where its rates are scaled by $\e^{-1}$ with $\e \ll 1$.
Based on similar assumptions to \cite{FK16b} we show that the corresponding Fokker-Planck equation for the coupled model (system with environment) is well-posed.
Moreover, we give a sufficient condition for the environment process to be ergodic (see Theorem \ref{TH:ERGODICITY}). 
Roughly speaking, the stochastic averaging principle asserts that the dynamics of the coupled model
may, in the scaling regime $\e \to 0$, be accurately described by an one-component dynamics with rates obtained by averaging 
the birth-and-death rates of the system w.r.t. the invariant measure of the environment. 
Such scheme is well-known in the physical literature and falls into the class of 
Markovian limits (see \cite{S88}).

Such kind of problems are well-developed in the framework of stochastic differential equations (see e.g. \cite{ET86},\cite{K92},\cite{PINSKY},\cite{SHS}).
The environments are typically relatively simple processes and the system consists of finitely many particles. 
Recently, we have established in \cite{FK16} the stochastic averaging principle for a system consisting of finitely many particles
evolving in an infinite particle environment in equilibrium.
The aim of this work is to extend this result to the case where both, the system and environment, are infinite particle systems
and, moreover, the environment is not assumed to be in equilibrium.

Since existence of a Markov process or an analysis of the backward Kolmogorov equation in the cases we consider is absent, 
we cannot directly apply the classical theory.
Having in mind that solutions to the Fokker-Planck equations are constructed in terms of correlation functions,
our main idea is to reformulate the problem in terms of correlation functions and then seek to apply abstract semigroup methods such as \cite{KURTZ73}.
One important difficulty in this approach is that, in contrast to finite particle systems, we cannot work with integrable correlation functions,
i.e. correlation functions corresponding to infinite particle systems are typically only Ruelle bounded.
The collection of all correlation functions then forms a cone in a weighted $L^{\infty}$-space 
where by the Lotz Theorem \cite{L85} any semigroup on such a space cannot by strongly continuous.
For this purpose we first study the pre-dual problem on a proper $L^1$-space describing the evolution of so-called quasi-observables.
Here we may apply \cite{KURTZ73} and then deduce our desired result on correlation functions by duality.

This work is organized as follows.
Our main results are formulated and discussed in the next section. 
For this purpose we introduce some notation used throughout this work.
Then we describe in some detail the environment process, the Fokker-Planck equation in Theorem \ref{TH:FPE}
and give a sufficient condition for the ergodicity of the environment in Theorem \ref{TH:ERGODICITY}. 
Afterwards we briefly discuss the system and the limiting process. The stochastic averaging principle is stated in Theorem \ref{TH:MAIN}.
Particular examples which show how this result can be applied are given in the third section.
The fourth section is devoted to the proofs of Theorem \ref{TH:FPE} and Theorem \ref{TH:ERGODICITY}.
Finally, based on the results of section four, a proof of the main result is given in section five.

\section{Statement of the results}

\subsection{Harmonic analysis on the configuration space}
The following is mainly based on \cite{KK02}.
Each particle is completely described by its position $x \in \R^d$ and the corresponding (one-type) configuration space is
\[
 \Gamma = \{ \gamma \subset \R^d \ | \ |\gamma \cap \Lambda| < \infty \text{ for all compacts } \Lambda \subset \R^d\},
\]
where $|\Lambda|$ denotes the number of elements in the set $\Lambda$.
It is well-known that $\Gamma$ is a Polish space w.r.t. the smallest topology such that 
$\gamma \longmapsto \sum_{x \in \gamma}f(x)$ is continuous for any continuous function $f$ having compact support (see \cite{KK06}).
Let $\mathcal{B}(\Gamma)$ be the Borel-$\sigma$-algebra on $\Gamma$. It is the smallest $\sigma$-algebra such that
$\Gamma \ni \gamma \longmapsto |\gamma \cap \Lambda|$ is measurable for any compact $\Lambda \subset \R^d$.

The space of finite configurations is defined by $\Gamma_0 := \{ \eta \subset \R^d \ | \ |\eta| < \infty\}$.
It is equipped with the smallest $\sigma$-algebra such that $\Gamma_0 \ni \eta \longmapsto |\eta \cap \Lambda|$
is measurable for any compact $\Lambda \subset \R^d$. 
We define a  measure $\lambda$ on $\Gamma_0$ by the relation
\[
 \int \limits_{\Gamma_0}G(\eta)d\lambda(\eta) = G(\emptyset) + \sum \limits_{n=1}^{\infty}\frac{1}{n!}\int\limits_{\R^{dn}}G(\{x_1,\dots, x_n\})dx_1\dots dx_n,
\]
where $G$ is any non-negative measurable function.
Let $B_{bs}(\Gamma_0)$ be the space of all bounded measurable functions with bounded support, i.e.
$G \in B_{bs}(\Gamma_0)$ iff $G$ is bounded and there exists a compact $\Lambda \subset \R^d$ and $N \in \N$ with $G(\eta) = 0$, whenever $\eta \cap \Lambda^c \neq \emptyset$ or
$|\eta| > N$. The $K$-transform is, for $G \in B_{bs}(\Gamma_0)$, defined by
\begin{align}\label{EQ:12}
 (KG)(\gamma) := \sum \limits_{\eta \Subset \gamma}G(\eta), \ \ \gamma \in \Gamma,
\end{align}
where $\Subset$ means that the sum is taken over all finite subsets of $\gamma$.
Let $\mu$ be a probability measure on $\Gamma^-$ with finite local moments, i.e. $\int_{\Gamma^-}|\gamma^- \cap \Lambda|^nd\mu(\gamma^-) < \infty$
for all compacts $\Lambda$ and $n \geq 1$. The correlation function $k_{\mu}: \Gamma_0 \longrightarrow \R_+$ is defined by the relation
\begin{align}\label{EQ:14}
 \int \limits_{\Gamma}(KG)(\gamma)d\mu(\gamma) = \int \limits_{\Gamma_0}G(\eta)k_{\mu}(\eta)d\lambda(\eta), \ \ G \in B_{bs}(\Gamma_0).
\end{align}
The Poisson measure $\pi_{z}$ with intensity measure $zdx$, $z > 0$, is main guiding example.
It is uniquely determined by the relation
$\pi_z( \{ \gamma \in \Gamma \ | \ |\gamma \cap \Lambda| = n \}) = \frac{m(\Lambda)^n z^n}{n!}$ 
where $m(\Lambda)$ denotes the Lebesgue measure of $\Lambda$.
Then $\pi_{z}$ has correlation function $k_{\pi_{z}}(\eta) = z^{|\eta|}$. 

Below we briefly describe how this notations are extended to the two-component case.
Namely, let $\Gamma^2 := \Gamma^+ \times \Gamma^-$ be equipped with the product topology where $\Gamma^{\pm}$ are two identical copies of $\Gamma$.
We let $\gamma = (\gamma^+,\gamma^-) \in \Gamma^2$, where $\gamma^+$ describes the particles of the system
and $\gamma^-$ the particles of the environment, respectively.
Similarly let $\Gamma_0^2 := \Gamma_0^+ \times \Gamma_0^-$, $\eta := (\eta^+, \eta^-)$ and $|\eta| := |\eta^+| + |\eta^-|$.
Define $G \in B_{bs}(\Gamma_0^2)$ iff $G$ is bounded and measurable
and there exists a compact $\Lambda \subset \R^d$ and $N \in \N$ such that $G(\eta) = 0$, whenever $|\eta| > N$ or $\eta^{\pm} \cap \Lambda^c \neq \emptyset$.
The $K$-transform is, for $G \in B_{bs}(\Gamma_0^2)$, defined by $(\mathbbm{K}G)(\gamma) := \sum_{\eta \Subset \gamma}G(\eta)$,
where $\xi \subset \eta$ and $\eta \Subset \gamma$ are defined component-wise.
Let $\mu$ be a probability measure on $\Gamma^2$ (=: state) with finite local moments, 
i.e. $\int_{\Gamma^2}|\gamma^- \cap \Lambda|^n |\gamma^+ \cap \Lambda|^n d\mu(\gamma) < \infty$ for all compacts $\Lambda$ and $n \geq 1$.
As before we define the correlation function $k_{\mu}$ by
\[
 \int \limits_{\Gamma^2}\mathbbm{K}G(\gamma)d \mu(\gamma) = \int \limits_{\Gamma_0^2}G(\eta)k_{\mu}(\eta)d \lambda(\eta), \ \ G \in B_{bs}(\Gamma_0^2).
\]
At this point it is worth to mention that not every non-negative function $k$ on $\Gamma_0^2$ is the correlation function of some state $\mu$.
It is necessary and sufficient that $k(\emptyset) = 1$ and that $k$ is positive definite in the sense of Lenard (see \cite{L75a}, \cite{L75b}).

\subsection{The environment}
Particles, in the framework of spatial birth-and-death processes, may randomly disappear and new particles may appear in the configuration $\gamma^- \in \Gamma^-$.
Death of a particle $x \in \gamma^-$ is described by the death rate $d^-(x,\gamma^-) \geq 0$. 
Similarly, $b^-(x,\gamma^-) \geq 0$ describes the birth rate and distribution of a new particle $x \in \R^d \backslash \gamma^-$.
In this work we assume that the environment is described by a Markov operator (formally) given by
\begin{align}\label{MCSL:21}
 (L^-F)(\gamma^-) = \sum \limits_{x \in \gamma^-}d^-(x,\gamma^- \backslash x)(F(\gamma^- \backslash x) - F(\gamma^-)) + \int \limits_{\R^d}b^-(x,\gamma^-)(F(\gamma^- \cup x) - F(\gamma^-))dx
\end{align}
where $F \in \mathcal{FP}(\Gamma^-) := K(B_{bs}(\Gamma_0^-))$.
For simplicity of notation we have let $\gamma^- \backslash x, \gamma^- \cup x$ stand for $\gamma^- \backslash \{x\}, \gamma^- \cup \{x\}$.
Note that in general we cannot expect that $d^-,b^-$ are well-defined for all $\gamma^- \in \Gamma^-$ and $x \in \R^d$.
Below we discuss our assumptions on the birth-and-death rates.
\begin{enumerate}
 \item[(E1)] There exist measurable functions $D^-, B^-: \R^d \times \Gamma_0^- \longrightarrow \R$ such that
 \begin{align}\label{EQ:20}
  d^-(x,\gamma^-) = \sum \limits_{\eta^- \Subset \gamma^-}D^-(x,\eta^-), \qquad b^-(x,\gamma^-) = \sum \limits_{\eta^- \Subset \gamma^-}B^-(x,\eta^-), \ \ x \in \R^d, \ \ \gamma^- \in \Gamma^-,
 \end{align}
 Moreover there exist constants $A > 0$, $N \in \N$ and $\nu \geq 0$ such that
 \begin{align*}
  b^-(x,\eta^-) + d^-(x,\eta^-) \leq A(1+|\eta^-|)^{N}e^{\nu|\eta^-|}, \ \ \eta^- \in \Gamma_0^-, \ x \in \R^d.
 \end{align*}
 \item[(E2)] There exists $C_- > 0$ and $a_- \in (0,2)$ such that
  \[
   c_-(\eta^-) \leq a_- M_-(\eta^-), \ \ \eta^- \in \Gamma_0^-
  \]
  holds, where $M_-(\eta^-) = \sum_{x \in \eta^-}d^-(x,\eta^- \backslash x)$ and
  \begin{align*}
   c_-(\eta^-) := &\ \sum \limits_{x \in \eta^-}\int \limits_{\Gamma_0^-}\left| \sum \limits_{\zeta^- \subset \eta^- \backslash x}D^-(x,\xi^- \cup \zeta^-)\right|C_-^{|\xi^-|}d \lambda(\xi^-)
  \\ &\ \ \ + \frac{1}{C_-}\sum \limits_{x \in \eta^-}\int \limits_{\Gamma_0^-}\left| \sum \limits_{\zeta^- \subset \eta^- \backslash x}B^-(x,\xi^- \cup \zeta^-)\right|C_-^{|\xi^-|}d \lambda(\xi^-).
 \end{align*}
\end{enumerate}
A priori it is not clear that the sums in \eqref{EQ:20} are absolutely convergent, it is part of the assertion of Lemma \ref{ENV:LEMMA00}.
For the last condition let $(R_{\delta})_{\delta > 0} \subset L^1(\R^d) \cap C(\R^d)$ 
with $0 < R_{\delta} \leq 1$ and $R_{\delta}(x) \nearrow 1$ for all $x$ as $\delta \to 0$. 
\begin{enumerate}
 \item[(E3)] For all $\delta > 0$ there exists a measurable function $V_{\delta}: \Gamma_0^- \longrightarrow \R_+$ and constants $c_{\delta},\e_{\delta} > 0$ such that 
 \[
  (L_{\delta}^-V_{\delta})(\eta^-) \leq c_{\delta}( 1 + V_{\delta}(\eta^-)) - \e_{\delta} D_{\delta}^-(\eta^-), \ \ \eta^- \in \Gamma_0^-
 \]
 where $L_{\delta}^-$ is obtained from $L^-$ with $b^-(x,\cdot)$ replaced by $R_{\delta}(x)b^-(x,\cdot)$ and 
 \[
  D_{\delta}^-(\eta^-) = M_-(\eta^-) + \int \limits_{\R^d}R_{\delta}(x)b^-(x,\eta^-)dx.
 \]
\end{enumerate}
Let us give some additional comments on this assumptions.
\begin{remark}
 \begin{enumerate}
  \item[(i)] Here $d^-(x,\emptyset) = D^-(x,\emptyset)$ describes the constant mortality (without interactions)
  whereas $D^-(x,\{y\})$ takes pair interactions (competition of $x$ with $y$) into account. General $n$-point interactions are
  described by $D^-(x,\eta^-)$ where $|\eta^-| = n$.
  \item[(ii)] For many particular models, see the next section, the function $c_-$ can be computed explicitly.
  Condition (E2) is satisfied, provided $d^-(x,\emptyset)$ is large enough, i.e. it describes some sort of ''high-mortality regime''.
  \item[(iii)] Condition (E3) asserts the existence of some sort of Lyapunov function for the generator $L_{\delta}^-$.
  The latter one describes a birth-and-death Markov process on $\Gamma_0^-$ and under the given condition this process is conservative. 
  A simpler sufficient condition for (E3) is given in the next section;
  e.g. if $b^-(x,\eta^-) \leq A(1+|\eta^-|)$ for some constant $A > 0$, then (E3) holds.
 \end{enumerate}
\end{remark}
We study dynamics of the environment in terms of one-dimensional distributions, i.e. solutions to the Fokker-Planck equation
\begin{align}\label{EQ:06}
 \frac{d}{dt}\int \limits_{\Gamma^-}F(\gamma^-)d \mu_t^-(\gamma^-) = \int \limits_{\Gamma^-}L^-F(\gamma^-)d \mu_t^-(\gamma^-), \ \ \mu_t^-|_{t= 0} = \mu_0^-, \ \ t \geq 0.
\end{align}
The latter one is analyzed in the class $\mathcal{P}_{C_-}$, where 
$\mu \in \mathcal{P}_{C_-}$ if and only if the correlation function $k_{\mu}$ exists and satisfies for some constant $A(\mu) > 0$ the so-called Ruelle bound 
\begin{align}\label{RUELLE:BOUND}
 k_{\mu}(\eta^-) \leq A(\mu)C_-^{|\eta^-|}, \ \ \eta^- \in \Gamma_0^-.
\end{align}
It is worth to mention that under condition \eqref{RUELLE:BOUND} the correlation function $k_{\mu}$ uniquely determines the state $\mu$
(see \cite{L75a}). The next lemma shows that $L^-$ is well-defined.
\begin{lemma}\label{ENV:LEMMA00}
 Suppose that conditions (E1) and (E2) are satisfied. Then
 \begin{enumerate}
  \item[(a)] For any $\mu \in \mathcal{P}_{C_-}$ and $x \in \R^d$ we have
  \begin{align*}
   \int \limits_{\Gamma^-}\left(d^-(x,\gamma^-) + b(x,\gamma^-)\right)d\mu(\gamma^-) \leq \max\{1, C_- \} A(\mu)a_- d^-(x,\emptyset).
  \end{align*}
  In particular, the sums in \eqref{EQ:20} are absolutely convergent for $\mu$-a.a. $\gamma^- \in \Gamma^-$ and all $x \in \R^d$.
  \item[(b)] We have $L^-F \in L^1(\Gamma^-, d\mu)$ for all $\mu \in \mathcal{P}_{C_-}$ and $F \in \mathcal{FP}(\Gamma^-)$.
 \end{enumerate}
\end{lemma}
A proof is given in section four.
Following the general scheme described in \cite{KKM08}, we study solutions to \eqref{EQ:06} in terms of correlation functions $(k_t)_{t \geq 0}$.
The latter ones should (at least formally) satisfy a Markov analogue of the BBGKY-hierarchy known from physics (see equation \eqref{EQ:01}).
Motivated by \eqref{RUELLE:BOUND} we denote by $\K_{C_-}$ the Banach space of essentially bounded functions $k$ on $\Gamma_0^-$ with norm
$\| k \|_{\K_{C_-}} = \esssup_{\eta \in \Gamma_0^-} | k(\eta^-)|C_-^{-|\eta^-|}$.
\begin{theorem}\label{TH:FPE}
 Suppose that conditions (E1) -- (E3) are satisfied. 
 Then for each $\mu_0^- \in \mathcal{P}_{C_-}$ 
 there exists a family of states $(\mu_t^-)_{t \geq 0} \subset \mathcal{P}_{C_-}$ with the following properties
 \begin{enumerate}
  \item[(a1)] $t \longmapsto \int_{\Gamma^-}(L^-F)(\gamma^-)d\mu_t^-(\gamma^-)$ is continuous for all $F \in \mathcal{FP}(\Gamma^-)$.
  \item[(a2)] $\| k_{\mu_t} \|_{\K_{C_-}} \leq \| k_{\mu_0} \|_{\K_{C_-}}$ for all $t \geq 0$.
  \item[(a3)] $t \longmapsto \int_{\Gamma^-}F(\gamma^-)d\mu_t^-(\gamma^-)$ is continuously differentiable and \eqref{EQ:06} holds for all $F \in \mathcal{FP}(\Gamma^-)$.
 \end{enumerate}
 Moreover, this solution is unique among all $(\nu_t^-)_{t \geq 0} \subset \mathcal{P}_{C_-}$ which satisfy
 \begin{enumerate}
  \item[(b1)] $t \longmapsto \int_{\Gamma^-}(L^-F)(\gamma^-)d\nu_t^-(\gamma^-)$ is locally integrable for all $F \in \mathcal{FP}(\Gamma^-)$.
  \item[(b2)] $\sup_{t \in [0,T]}\| k_{\nu_t} \|_{\K_{C_-}} < \infty$ holds for all $T > 0$.
  \item[(b3)] $t \longmapsto \int_{\Gamma^-}F(\gamma^-)d\nu_t^-(\gamma^-)$ is absolutely conditions and \eqref{EQ:06} holds 
  for all $F \in \mathcal{FP}(\Gamma^-)$ and a.a. $t \geq 0$.
 \end{enumerate}
\end{theorem}
Classical solutions to the BBGKY-hierarchy (see \eqref{EQ:01}) with general rates have been first obtained in \cite{FKK12},\cite{FKK15}
in the class of functions obeying \eqref{RUELLE:BOUND}.
The construction given in this work relies on the same idea, but we use suitable perturbation theory for substochastic semigroups on weighted $L^1$-spaces instead.
At this point condition (E2) is used to show that certain operators involved in the construction are relatively bounded.

In general a solution to the BBGKY-hierarchy does not need to determine uniquely an evolution of states.
It is necessary and sufficient to show that such a solution is positive definite in the sense of Lenard (see \eqref{EQ:19} for the definition).
Such a property was shown under more stringent assumptions in \cite{FK16b}. Namely an additional technical assumption (stronger then (E2))
was imposed on $c_-$; the initial condition $\mu_0^-$ was assumed to belong to some strictly smaller subspace $\mathcal{P}' \subset \mathcal{P}_{C_-}$;
and finally condition (E3) was replaced by some more technical condition.
Since the main ideas of the proof are the same as in \cite{FK16b}, we do not give a full proof.
Instead we outline the most important steps in the first part of section four.
\begin{remark}
 It is not difficult to adapt this result for a spatial birth-and-death process on $\Gamma^2$. This will be used, without proof, later on.
\end{remark}
Recently, we have studied in \cite{F16WR} exponential ergodicity for a two-component Glauber-type process 
where particles are allowed to change their type at certain multiplicative rates. 
In this work we give a general result applicable for a birth-and-death process with Markov operator \eqref{MCSL:21}.
\begin{theorem}\label{TH:ERGODICITY}
 Suppose that conditions (E1) -- (E3) are satisfied. Moreover, assume that
 \begin{align}\label{EQ:05}
  \inf \left\{  \sum \limits_{x \in \eta^-}d^-(x,\eta^- \backslash x) \ \bigg | \ \emptyset \neq \eta^- \in \Gamma_0^- \right\} > 0.
 \end{align}
 Then there exists  $\mu_{\mathrm{inv}} \in \mathcal{P}_{C_-}$ with correlation function $k_{\mathrm{inv}}$ such that:
 \begin{enumerate}
  \item[(a)] $\mu_{\mathrm{inv}}$ is the unique stationary solution to \eqref{EQ:06}.
  \item[(b)] There exist constants $A, \lambda > 0$ such that for any $\mu_0^- \in \mathcal{P}_{C_-}$ it holds that
  \begin{align*}
  \| k_{\mu_t^-} - k_{\mu_{\mathrm{inv}}} \|_{\K_{C_-}} \leq A e^{-\lambda t}\| k_{\mu_0^-} - k_{\mu_{\mathrm{inv}}}\|_{\K_{C_-}}, \ \ t \geq 0
  \end{align*}
  where $(\mu_t^-)_{t \geq 0} \subset \mathcal{P}_{C_-}$ is the unique solution to \eqref{EQ:06}.
 \end{enumerate}
\end{theorem}
The proof uses some arguments similar to \cite{F16WR}, but since we work with more general conditions additional technical steps have to be done.
A proof of this statement is given in section four.
Existence of the invariant measure is obtained from some type of generalized Kirkwood-Salsburg equation (see also \cite{FKK12}).
Ergodicity is then deduced by classical spectral theory, i.e. we show that the generator for the dynmaics has a spectral gap.
\begin{remark}
 \begin{enumerate}
  \item[(i)] If $b^-(x,\emptyset) = 0$, then $\mu_{\mathrm{inv}} = \delta_{\emptyset}$, i.e. the population gets extinct with exponential speed.
  \item[(ii)] If $b^-(x,\emptyset) > 0$, then $\mu_{\mathrm{inv}} \neq \delta_{\emptyset}$.
 \end{enumerate}
\end{remark}
Particular examples such as the Sourgailis model with invariant measure $\pi_{z}$ and the Glauber dynamics with a Gibbs measure as the invariant
measure have been studied in \cite{F11}, \cite{KKM10}.
The aggregation model considered in \cite{FKKZ14} is a particular example where $d^-(x,\emptyset)$ is bounded away from zero, but condition \eqref{EQ:05} does not hold.

\subsection{The system}
Dynamics for the system is described by the birth-and-death rates $d^+(x,\gamma)$ and $d^+(x,\gamma)$
where the additional dependence on the parameter $\gamma^-$ takes interactions of the system with its environment into account.
The Markov operator is formally given by
\begin{align*}
 (L^+F)(\gamma) = &\ \sum \limits_{x \in \eta^+}d^+(x,\gamma^+ \backslash x, \gamma^-)(F(\gamma^+ \backslash x, \gamma^-) - F(\gamma^+, \gamma^-))
 \\ \notag &+ \int \limits_{\R^d}b^-(x,\gamma^+,\gamma^-)(F(\gamma^+ \cup x, \gamma^-) - F(\gamma^+, \gamma^-))d x
\end{align*}
where $F \in \mathcal{FP}(\Gamma^2) = \mathbb{K}(B_{bs}(\Gamma_0^2))$.
Since no confusion may arise we also denote by $\lambda$ the measure $\lambda \otimes \lambda$.
Similarly to (E1) -- (E3) we impose the following conditions:
\begin{enumerate}
 \item[(S1)] There exist measurable functions $D^+, B^+: \R^d \times \Gamma_0^2 \longrightarrow \R$ such that
 \[
  d^+(x,\gamma) = \sum \limits_{\eta \Subset \gamma}D^+(x,\eta), \qquad b^+(x,\gamma) = \sum \limits_{\eta \Subset \gamma}B^+(x,\eta).
 \]
 Moreover there exist constants $A > 0$, $N \in \N$ and $\nu \geq 0$ such that
 \begin{align*}
  b^+(x,\eta) + d^+(x,\eta) \leq A(1+|\eta|)^{N}e^{\nu|\eta|}, \ \ \eta \in \Gamma_0^2, \ x \in \R^d.
 \end{align*}
 \item[(S2)] There exists $C_+ >$ and $a_+ \in (0,2)$ such that
 \[
  c_{+}(\eta) \leq a_{+} M_{+}(\eta), \ \ \eta \in \Gamma_0^2
 \]
 holds where $M_{+}(\eta) = \sum_{x \in \eta^+}d^+(x, \eta^+ \backslash x, \eta^-)$ and
 \begin{align*}
 c_+(\eta) :=
     &\ \sum \limits_{x \in \eta^+}\int \limits_{\Gamma_0^2}\left | \sum \limits_{\bfrac{\zeta^+ \subset \eta^+ \backslash x}{\zeta^- \subset \eta^-}}D^+(x,\zeta^+ \cup \xi^+,\zeta^- \cup \xi^-)\right| C_+^{|\xi^+|}C_-^{|\xi^-|}d \lambda(\xi)
  \\ &+ \frac{1}{C_+}\sum \limits_{x \in \eta^+} \int \limits_{\Gamma_0^2}\left | \sum \limits_{\bfrac{\zeta^+ \subset \eta^+ \backslash x}{\zeta^- \subset \eta^-}}B^+(x,\zeta^+ \cup \xi^+,\zeta^- \cup \xi^-)\right | C_+^{|\xi^+|}C_-^{|\xi^-|}d \lambda(\xi).
 \end{align*} 
 \item[(S3)] Take $(R_{\delta})_{\delta > 0}$ as in (E3).
 For any $\delta > 0$ there exists a measurable function $V_{\delta}: \Gamma_0^2 \longrightarrow \R_+$ and constants $c_{\delta}, \e_{\delta} > 0$ such that
 \[
  (L_{\delta}^+V_{\delta})(\eta) \leq c_{\delta}(1 + V_{\delta}(\eta)) - \e_{\delta} D_{\delta}^+(\eta), \ \ \eta \in \Gamma_0^2
 \]
 where $L_{\delta}^+$ is given by $L^+$ with $b^+(x,\cdot)$ replaced by $R_{\delta}(x)b^+(x,\cdot)$ and
 \[
  D_{\delta}^+(\eta) = M_+(\eta) + \int \limits_{\R^d}R_{\delta}(x)b^+(x,\eta)dx.
 \]
\end{enumerate}
Let $\mathcal{P}_{C}$ be the space of all $\mu$ such that $\mu \in \mathcal{P}_C$ iff its correlation function $k_{\mu}$ exists and satisfies the Ruelle bound
\[
 k_{\mu}(\eta) \leq A(\mu)C_+^{|\eta^+|}C_-^{|\eta^-|}, \ \ \eta \in \Gamma_0^2.
\]
Arguing similarly to the proof of Lemma \ref{ENV:LEMMA00} we see that 
\[
 \sup \limits_{x \in \R^d}\int \limits_{\Gamma^2}\left( d^+(x,\gamma) + b^+(x,\gamma)\right)d\mu(\gamma) < \infty,
\]
$L^+F \in L^1(\Gamma^2,d\mu)$ for $\mu \in \mathcal{P}_{C}$ and $F \in \mathcal{FP}(\Gamma^2)$
and the following analogue of Theorem \ref{TH:FPE} holds.
\begin{theorem}\label{TH:FPESYSTEM}
 Suppose that (S1) -- (S3) are satisfied.
 Then for any $\mu_0 \in \mathcal{P}_{C}$ there exists a unique solution $(\mu_t)_{t \geq 0} \subset \mathcal{P}_{C}$ to
 \[
  \frac{d}{dt}\int \limits_{\Gamma^2}F(\gamma)d\mu_t(\gamma) = \int \limits_{\Gamma^2}L^+F(\gamma)d\mu_t(\gamma), \ \ \mu_t|_{t=0} = \mu_0, \ \ t \geq 0, \ \ F \in \mathcal{FP}(\Gamma^2).
 \]
\end{theorem}
This describes the evolution of the system in the presence of a stationary environment fixed with the choice of $\mu_0$.

\subsection{The averaged process}
Here and below we assume that conditions (E1) -- (E3), (S1) -- (S3) and \eqref{EQ:05} are satisfied.
Consider the averaged Markov operator given by
\begin{align*}
 (\overline{L}F)(\gamma^+) = \sum \limits_{x \in \gamma^+}\overline{d}(x,\gamma^+ \backslash x)(F(\gamma^+ \backslash x) - F(\gamma^+))
 + \int \limits_{\R^d}\overline{b}(x,\gamma^+)(F(\gamma^+ \cup x) - F(\gamma^+))d x
\end{align*}
where $F \in \mathcal{FP}(\Gamma^+)$. The birth-and-death rates are obtained by integration w.r.t. to the invariant measure of the environment, i.e.
\begin{align}
 \label{MCSL:23} \overline{d}(x,\gamma^+) &:= \int \limits_{\Gamma}d^+(x,\gamma)d \mu_{\mathrm{inv}}(\gamma^-) = \sum \limits_{\eta^+ \Subset \gamma^+}\overline{D}(x,\eta^+),
 \\ \label{MCSL:24} \overline{b}(x,\gamma^+) &:= \int \limits_{\Gamma}b^+(x,\gamma)d \mu_{\mathrm{inv}}(\gamma^-) = \sum \limits_{\eta^+ \Subset \gamma^+}\overline{B}(x,\eta^+),
\end{align}
where $\overline{D}(x,\eta^+) = \int_{\Gamma_0^-}D^+(x,\eta^+,\eta^-)k_{\mathrm{inv}}(\eta^-)d\lambda(\eta^-)$ and
$\overline{B}(x,\eta^+) = \int_{\Gamma_0^-}B^+(x,\eta^+,\eta^-)k_{\mathrm{inv}}(\eta^-)d\lambda(\eta^-)$.
In order to proceed it is necessary to assume that similar conditions to (E1) -- (E3) with $d^-,b^-$ replaced by $\overline{d},\overline{b}$ are satisfied,
i.e. we assume:
\begin{enumerate}
 \item[(AV1)] There exist $A > 0$, $N \in \N$ and $\nu \geq 0$ such that
  \[
   \overline{d}(x,\eta^+)\leq A(1 + |\eta^+|)^N e^{\nu|\eta^+|}
  \]
  holds for all $x \in \R^d$ and $\eta^+ \in \Gamma_0^+$.
 \item[(AV2)] There exists a constant $\overline{a} \in (0,2)$ such that
 \begin{align*}
  \overline{c}(\eta^+) \leq \overline{a} \overline{M}(\eta^+), \ \ \eta^+ \in \Gamma_0^+
 \end{align*}
 holds where $\overline{M}(\eta^+) := \sum_{x \in \eta^+}\overline{d}(x,\eta^+ \backslash x)$ and 
 $\overline{c}(\cdot)$ is defined as $c_-(\cdot)$ with $D^-,B^-$ replaced by $\overline{D}, \overline{B}$ and $C_-$ replaced by $C_+$.
 \item[(AV3)] Take $(R_{\delta})_{\delta > 0}$ as in (E3). For any $\delta > 0$ there exists a measurable function $V_{\delta}: \Gamma_0^+ \longrightarrow \R_+$
 and constants $c_{\delta},\e_{\delta} > 0$ such that
 \[
  (\overline{L}_{\delta}V_{\delta})(\eta^+) \leq c_{\delta}(1 + V_{\delta}(\eta^+)) - \e_{\delta} \overline{M}_{\delta}(\eta^+), \ \ \eta^+ \in \Gamma_0^+
 \]
 where $\overline{L}_{\delta}$ is given by $\overline{L}$ with $\overline{b}(x,\cdot)$ replaced by $R_{\delta}(x)\overline{b}(x,\cdot)$ and
 \[
  \overline{D}_{\delta}(\eta^+) = \overline{M}(\eta^+) + \int \limits_{\R^d}R_{\delta}(x)\overline{b}(x,\eta^+)dx.
 \]
\end{enumerate}
One would expect that (S2), (S3) together with \eqref{MCSL:23} and \eqref{MCSL:24} already imply (AV2) and (AV3).
Unfortunately we could not show that this is, indeed, the case.
Particular examples considered in the next section, however, show that conditions (AV1) -- (AV3) are merely a restriction.

\subsection{Stochastic averaging principle}
We consider the Fokker-Planck for the system and environment given by the Markov operator $L^+ + \frac{1}{\e}L^-$.
Here $L^-$ is extended to functions $F: \Gamma^2 \longrightarrow \R$ by its action only on the variable $\gamma^-$.
For a given state $\mu \in \mathcal{P}_{C}$ the marginal on the first component $\mu^+$ is defined by
\begin{align*}
 \int \limits_{\Gamma^+}F(\gamma^+)d \mu^+(\gamma^+) = \int \limits_{\Gamma^2}F(\gamma^+)d \mu(\gamma^+,\gamma^-).
\end{align*}
The following is our main result.
\begin{theorem}\label{TH:MAIN}
 Suppose that conditions (E1) -- (E3), (S1) -- (S3), (AV1) -- (AV3) and \eqref{EQ:05} are satisfied. Then the following assertions hold:
 \begin{enumerate}
  \item[(a)] For each $\mu_0 \in \mathcal{P}_{C}$ and each $\e > 0$ there exists a unique solution $(\mu_t^{\e})_{t \geq 0} \subset \mathcal{P}_{C}$ to 
  \begin{align*}
   \frac{d}{dt}\int \limits_{\Gamma^2}F(\gamma)d \mu_{t}^{\varepsilon}(\gamma) = \int \limits_{\Gamma^2}\left( L^+ + \frac{1}{\varepsilon}L^-\right)F(\gamma)d \mu_t^{\varepsilon}(\gamma), \ \ \mu_t^{\e}|_{t=0} = \mu_0, \ \ t \geq 0.
  \end{align*}
  \item[(b)] For each $\overline{\mu}_0 \in \mathcal{P}_{C_+}$ there exists a unique solution $(\overline{\mu}_t)_{t \geq 0} \subset \mathcal{P}_{C_+}$ to
  \begin{align}\label{FPE:AVERAGED}
   \frac{d}{dt}\int \limits_{\Gamma^+}F(\gamma^+)d\overline{\mu}_t(\gamma^+) = \int \limits_{\Gamma^+}(\overline{L}F)(\gamma^+)d\overline{\mu}_t(\gamma^+), \ \ \overline{\mu}_t|_{t=0} = \overline{\mu}_0, \ \ t \geq 0.
  \end{align}
  \item[(c)] For any $F \in \mathcal{FP}(\Gamma^+)$ we have
   \begin{align*}
    \int \limits_{\Gamma^2}F(\gamma^+)d \mu_t^{\varepsilon}(\gamma^+,\gamma^-) \longrightarrow \int \limits_{\Gamma^+}F(\gamma^+)d \overline{\mu}_t(\gamma^+), \ \ \varepsilon \to 0
   \end{align*}
   uniformly on compacts in $t \geq 0$ where $(\overline{\mu}_t)_{t \geq 0}$ is the unique solution to \eqref{FPE:AVERAGED} 
   with initial condition $\overline{\mu}_0 = \mu_0^+$.
 \end{enumerate}
\end{theorem}
Note that assertion (a) extends Theorem \ref{TH:FPESYSTEM} since here the environment does not need to be stationary.
This result also extends \cite{FK16} where the system was a birth-and-death process on $\Gamma_0^+$ and the environment was an equilibrium process on $\Gamma^-$.
Here the system and environment are both infinite particle systems and secondly the environment is only assumed to be ergodic; it does not need to be
in equilibrium.
In a forthcoming work we will use the results obtained in this work to extend the stochastic averaging principle 
for a model where conditions (E2), (S2) and (AV2) are relaxed.

\section{Examples}
Many models of interacting particle systems are based on translation invariant rates (see e.g. \cite{FKK15}). 
Such rates may result from an idealisation and simplification of the underlying physical model under which particular properties can be studied and observed.
However, biological systems related with the description of tumour growth should, due to their complex spatial structure, be modelled by 
space-inhomogeneous rates.
The particular choice of such rates is often based on ad-hoc assumptions and a deep understanding of the underlying nature of the dynamics involved.
Using the stochastic averaging principle, we show that such rates can be rigorously derived from the interaction with a Markovian environment.
The birth-and-death rates in the examples given below are build by relative energies
\[
 E_{\varphi}(x,\gamma^{\pm}) := \sum \limits_{y \in \gamma^{\pm}}\varphi(x-y), \ \ x \in \R^d, \ \ \gamma^{\pm} \in \Gamma,
\]
where $\varphi$ is a symmetric, non-negative and, integrable function. 

\subsection{Preliminaries}
We will frequently use the combinatorial relation
\begin{align}\label{IBP}
 \int \limits_{\Gamma_0^{\pm}}\sum \limits_{\xi^{\pm} \subset \eta^{\pm}}G(\xi^{\pm},\eta^{\pm} \backslash \xi^{\pm})d \lambda(\xi^{\pm}) = \int \limits_{\Gamma_0^{\pm}}\int \limits_{\Gamma_0^{\pm}}G(\xi^{\pm},\eta^{\pm})d \lambda(\xi^{\pm})d \lambda(\eta^{\pm})
\end{align}
provided one side of the equality is finite for $|G|$, cf. \cite{F16WR}. 
Moreover by \cite{F09} we have
\begin{align}\label{EQ:04}
 \lambda\left( \{ \eta^{\pm} \in \Gamma_0^{\pm} \ | \ \xi^{\pm} \cap \eta^{\pm} \neq \emptyset \} \right) = 0, \ \ \xi^{\pm} \in \Gamma_0^{\pm}
\end{align}
and $\lambda\otimes \lambda \left( \{ \eta \in \Gamma_0^2 \ | \ \eta^+ \cap \eta^- \neq \emptyset \} \right) = 0$.
For a given measurable function $g$ we have
\begin{align}\label{EQ:07}
 \sum \limits_{\xi^{\pm} \subset \eta^{\pm}}\prod \limits_{x \in \xi^{\pm}}g(x) = \prod \limits_{x \in \eta^{\pm}}(1+g(x))
\end{align}
and if $|g|$ is integrable, then 
\begin{align}\label{EQ:09}
 \int \limits_{\Gamma_0^{\pm}}\prod \limits_{x \in \eta^{\pm}}g(x)d\lambda(\eta^{\pm}) = \exp\left(\int \limits_{\R^d}g(x)dx\right).
\end{align}
Let us finally give a sufficient condition for (E3).
\begin{proposition}
 Suppose that for each $\delta > 0$ there exists $c_{\delta} > 0$ and $\e_{\delta} \in (0,1)$ such that
 \[
  \int \limits_{\R^d}R_{\delta}(x)b^-(x,\eta^-)dx \leq c_{\delta}(1+|\eta^-|) + \e_{\delta} \sum \limits_{x \in \eta^-}d^-(x,\eta^- \backslash x), \ \ \eta^- \in \Gamma_0^-.
 \]
 Then condition (E3) is satisfied.
 In particular, if $b^-(x,\eta^-) \leq A(1 + |\eta^-|)$ holds for some constant $A > 0$, then (E3) is satisfied.
\end{proposition}
\begin{proof}
 Write $c_{\delta} = \frac{2 c_{\delta}'}{1 + \e_{\delta}'} > 0$ and $\e_{\delta} = \frac{1 - \e_{\delta}'}{1 + \e_{\delta}'} \in (0,1)$
 where $c_{\delta}' = \frac{c_{\delta}}{1+\e_{\delta}}$ and $\e_{\delta}' = \frac{1 - \e_{\delta}}{1 + \e_{\delta}} \in (0,1)$. Then
 \[
  \int \limits_{\R^d}R_{\delta}(x)b^-(x,\eta^-)dx
  \leq 2 c_{\delta}'(1 + |\eta^-|) + (1 - \e_{\delta}')\sum \limits_{x \in \eta^-}d^-(x,\eta^- \backslash x) - \e_{\delta}' \int \limits_{\R^d}R_{\delta}(x)b^-(x,\eta^-)dx
 \]
 and hence we obtain for $V_{\delta}(\eta^-) = |\eta^-|$
 \begin{align*}
  (L_{\delta}^-V_{\delta})(\eta^-) &= - \sum \limits_{x \in \eta^-}d^-(x,\eta^- \backslash x) + \int\limits_{\R^d}R_{\delta}(x)b^-(x,\eta^-)dx
  \\ &\leq 2c_{\delta}'(1 + |\eta^-|) - \e_{\delta}'\sum \limits_{x \in \eta^-}d^-(x,\eta^- \backslash x) - \e_{\delta}' \int\limits_{\R^d}R_{\delta}(x)b^-(x,\eta^-)dx.
 \end{align*}
 \qed
\end{proof}
A similar statement can be shown for (S3) and (AV3).

\subsection{Two interacting Glauber dynamics}
The environment process is assumed to be a Glauber dynamics in the continuum with the birth-and-death rates
\begin{align}\label{GLAUBER:ENV}
 \begin{cases} d^-(x,\gamma) &= 1, 
 \\ b^-(x, \gamma) &= z^- \exp\left(- E_{\psi}(x,\gamma^-)\right) \end{cases}.
\end{align}
Such dynamics has been studied in \cite{FKKZ12}. The system process is another Glauber dynamics given by
\begin{align*}
 d^+(x,\gamma) &= 1,
 \\ b^+(x, \gamma) &= z^+ \exp\left( - E_{\phi^-}(x, \gamma^-)\right) \exp\left(- E_{\phi^+}(x,\gamma^+)\right).
\end{align*}
A similar model has been studied in \cite{FKK15WRMODEL} and \cite{F16WR}.
For an integrable function $f: \R^d \longrightarrow \R_+$ let
\begin{align*}
 \beta(f) := \int \limits_{\R^d}|e^{-f(x)} - 1|d x \in [0,\infty].
\end{align*}
Suppose that the following conditions are satisfied:
\begin{enumerate}
 \item[(a)] $\psi, \phi^{\pm} \geq 0$ are symmetric with $\beta(\psi), \beta(\phi^{\pm}) < \infty$ and $z^{\pm} > 0$.
 \item[(b)] There exist $C_{\pm} > 0$ such that 
  \begin{align}\label{GLAUBER:ENV:COND}
     z^- \exp\left( C_- \beta(\psi)\right) &< C_-,
  \\ \notag z^+ \exp\left(C_- \beta(\phi^-)\right) \exp\left(C_+ \beta(\phi^+)\right) &< C_+.
  \end{align}
\end{enumerate}
Let us prove that conditions (E1) -- (E3), (S1) -- (S3) and (AV1) -- (AV3) hold.
Condition (E1) is obvious with the choice
\[
 D^-(x,\eta^-) = 0^{|\eta^-|}, \ \ B^-(x,\eta^-) = z^- \prod \limits_{y \in \eta^-}\left( e^{-\psi(x-y)} - 1\right).
\]
Condition (E3) immediately follows from $b^-(x,\eta^-) \leq z^-$. Concerning condition (E2) we first observe that for $\xi^- \cap \zeta^- = \emptyset$ 
\[
 D^-(x,\xi^- \cup \zeta^-) = 0^{|\xi^-|}0^{|\zeta^-|}, \ \ B^-(x,\xi^- \cup \zeta^-) = z^- \prod \limits_{y \in \xi^-}\left( e^{-\psi(x-y)}-1\right)\prod \limits_{w \in \zeta^-}\left( e^{-\psi(x-y)}-1\right).
\]
In view of \eqref{EQ:04}, \eqref{EQ:07}, \eqref{EQ:09} we obtain
\begin{align*}
 c_-(\eta^-) &= |\eta^-| + \frac{z^-}{C_-} e^{C_-\beta(\psi)}\sum \limits_{x \in \eta^-}e^{-E_{\psi}(x,\eta^- \backslash x)} \leq |\eta^-|\left( 1 + \frac{z^-}{C_-} e^{C_-\beta(\psi)}\right).
\end{align*}
Hence (E2) follows from assumption (b). Similarly we show that assumptions (S1) -- (S3) hold.
Namely, condition (S3) follows from $b^+(x,\eta)\leq z^+$ and (S1) holds with
\[
 D^+(x,\eta) = 0^{|\eta|}, \ \ B^+(x,\eta) = z^+ \prod \limits_{y \in \eta^+}\left( e^{-\phi^+(x-y)} - 1\right)\prod \limits_{w \in \eta^-}\left( e^{-\phi^-(x-w)} - 1\right).
\]
Hence we obtain
\begin{align*}
 c_+(\eta) &= |\eta^+| + \frac{z^+}{C_+} e^{C_+\beta(\phi^+)}e^{C_-\beta(\phi^-)}\sum \limits_{x \in \eta^+}e^{-E_{\phi^+}(x,\eta^+ \backslash x)}e^{-E_{\phi^-}(x,\eta^-)}
 \\ &\leq |\eta^+|\left( 1 + \frac{z^+}{C_+} e^{C_+\beta(\phi^+)}e^{C_-\beta(\phi^-)}\right).
\end{align*}
In view of (b) condition (S2) holds.
The unique invariant measure for the environment is given by the Gibbs measure $\mu_{\mathrm{inv}}$ with activity $z^-$ and potential $\psi$. 
Moreover we have
\begin{align*}
 \overline{d}(x,\gamma^+) &= \int \limits_{\Gamma^-}d^+(x,\gamma^+, \gamma^-)d \mu_{\mathrm{inv}}(\gamma^-) = 1,
 \\ \overline{b}(x,\gamma^+) &= \int \limits_{\Gamma^-}b^+(x,\gamma^+, \gamma^-)d \mu_{\mathrm{inv}}(\gamma^-) 
  =  z^+\overline{\lambda}(x) \exp\left(- E_{\phi^+}(x,\gamma^+)\right).
\end{align*}
where $\overline{\lambda}(x) = \int_{\Gamma^-}e^{-E_{\phi^-}(x,\gamma^-)}d \mu_{\mathrm{inv}}(\gamma^-)$.
Arguing as above we see that (AV1) and (AV3) hold with
\[
 \overline{D}(x,\eta^+) = 0^{|\eta^+|}, \ \ \overline{B}(x,\eta^+) = z^+ \overline{\lambda}(x)\prod \limits_{y \in \eta^+}\left( e^{-\phi^+(x-y)} - 1 \right).
\]
Hence $\overline{c}(\eta^+) \leq \sum \limits_{x \in \eta^+}\left( 1 + \frac{z^+}{C_+}\overline{\lambda}(x)e^{C_+\beta(\phi^+)}\right) \leq |\eta^+|\left( 1 + \frac{z^+}{C_+}e^{C_+\beta(\phi^+)}e^{C_-\beta(\phi^-)}\right)$
implies (AV2).

\subsection{BDLP-dynamics in Glauber environment}
The environment is, as before, assumed to be a Glauber dynamics with birth-and-death rates \eqref{GLAUBER:ENV}.
The system is assumed to be a BDLP-process and it is assumed that the environment influences the system
due to additional competition via the potential $b^-$ and particles from the environment may create new individuals within the system.
More precisely we consider the rates
\begin{align*}
 d^+(x,\gamma) &= m^+ + \sum \limits_{y \in \gamma^+}a^-(x-y) + \sum \limits_{y \in \gamma^-}b^-(x-y),
 \\ b^+(x, \gamma) &= \sum \limits_{y \in \gamma^+}a^+(x-y) + \sum \limits_{y \in \gamma^-}b^+(x-y).
\end{align*}
Without influence of the environment ($b^{\pm} = 0$) the system is simply an one-component BDLP-process studied in \cite{BCFKKO14},\cite{FM04},\cite{KK16}.
In order to apply our results to this coupled model we make the following assumptions:
\begin{enumerate}
 \item[(a)] $\psi \geq 0$ is symmetric with $\beta(\psi) < \infty$ and there exists $C_- > 0$ such that \eqref{GLAUBER:ENV:COND} holds.
 \item[(b)] $m^+, z^- > 0$ and $a^{\pm}, b^{\pm} \geq 0$ are symmetric and integrable and bounded.
 \item[(c)] There exist $C_{+} > 0$, $\theta \in (0, C_+)$ and $b \geq 0$ such that
 \begin{align}\label{MCSL:72}
  \sum \limits_{x \in \eta^+} \sum \limits_{y \in \eta^+ \backslash x}a^+(x-y) \leq \theta \sum \limits_{x \in \eta^+}\sum \limits_{y \in \eta^+ \backslash x}a^-(x-y) + b|\eta^+|
 \end{align}
 holds and for some $\vartheta \in (0, C_+)$ we have
 \begin{align}
  \label{MCSL:73} \vartheta b^- &\geq b^+
  \\ \label{MCSL:75} m^+ &> C_-\| b^- \|_{L^1} + C_+\| a^- \|_{L^1} + \| a^+ \|_{L^1} + \frac{C_-}{C_+}\| b^+ \|_{L^1} + \frac{b}{C_+}.
 \end{align}
\end{enumerate}
\begin{remark}
 Condition \eqref{MCSL:72} states that $\theta a^- - a^+$ is a stable potential in the sense of Ruelle.
 Some sufficient conditions are given e.g. in \cite{KK16}.
\end{remark}
Let us show that (E1) -- (E3), (S1) -- (S3) and (AV1) -- (AV3) are satisfied.
First of all, previous example  shows that (E1) -- (E3) are satisfied.
Since $b^+(x,\eta) \leq |\eta^+|\| a^+ \|_{L^{\infty}} + |\eta^-|\| b^+\|_{L^{\infty}}$ it follows that (S3) holds.
Condition (S1) is satisfied for the choice
\begin{align*}
 D^+(x,\eta) &= m^+ 0^{|\eta|} + 0^{|\eta^-|}\1_{\Gamma_1^+}(\eta^+)\sum \limits_{y \in \eta^+}a^-(x-y) + 0^{|\eta^+|}\1_{\Gamma_1^-}(\eta^-)\sum \limits_{y \in \eta^-}b^-(x-y),
 \\ B^+(x,\eta) &= 0^{|\eta^-|}\1_{\Gamma_1^+}(\eta^+)\sum \limits_{y \in \eta^+}a^+(x-y) + 0^{|\eta^+|}\1_{\Gamma_1^-}(\eta^-)\sum \limits_{y \in \eta^-}b^+(x-y),
\end{align*}
where $\Gamma_1^{\pm} := \{ \eta^{\pm} \in \Gamma_0^{\pm} \ | \ |\eta^{\pm}| = 1 \}$. This gives for $\zeta^{\pm} \cap \xi^{\pm} = \emptyset$
\begin{align*}
 D^+(x,\zeta^+ \cup \xi^+, \zeta^- \cup \xi^-) &= m^+ 0^{|\xi|}0^{|\zeta|} 
 \\ &\ \ \ + 0^{|\zeta|}0^{|\xi^+|}\1_{\Gamma_1^-}(\xi^-)\sum \limits_{y \in \xi^-}b^-(x-y) +  0^{|\xi|}0^{|\zeta^+|}\1_{\Gamma_1^-}(\zeta^-)\sum \limits_{y \in \zeta^-}b^-(x-y)
 \\ &\ \ \ + 0^{|\xi^-|}0^{|\zeta|}\1_{\Gamma_1^+}(\xi^+)\sum \limits_{y \in \xi^+}a^-(x-y) + 0^{|\xi|}0^{|\zeta^-|}\1_{\Gamma_1^+}(\zeta^+)\sum \limits_{y \in \zeta^+}a^-(x-y) ,
 \\ B^+(x,\zeta^+ \cup \xi^+, \zeta^- \cup \xi^-) &= 0^{|\zeta|}0^{|\xi^-|}\1_{\Gamma_1^+}(\xi^+)\sum \limits_{y \in \xi^+}a^+(x-y) + 0^{|\xi|}0^{|\zeta^-|}\1_{\Gamma_1^+}(\zeta^+)\sum \limits_{y \in \zeta^+}a^+(x-y) 
 \\ &\ \ \ + 0^{|\zeta|}0^{|\xi^+|}\1_{\Gamma_1^-}(\xi^-)\sum \limits_{y \in \xi^-}b^+(x-y) + 0^{|\xi|}0^{|\zeta^+|}\1_{\Gamma_1^-}(\zeta^-)\sum \limits_{y \in \zeta^-}b^+(x-y).
\end{align*}
Using the definition of $\lambda$ we get
\begin{align*}
 c_+(\eta) &= m^+ |\eta^+| + C_-\| b^- \|_{L^1}|\eta^+| + C_+\| a^- \|_{L^1}|\eta^+| + \sum \limits_{x \in \eta^+}\sum \limits_{y \in \eta^+ \backslash x}a^-(x-y) + \sum \limits_{x \in \eta^+}\sum \limits_{y \in \eta^-}b^-(x-y)
 \\ &+ |\eta^+| \| a^+ \|_{L^1} + \frac{C_-}{C_+}\| b^+ \|_{L^1}|\eta^+| + \frac{1}{C_+}\sum \limits_{x \in \eta^+}\sum \limits_{y \in \eta^+ \backslash x}a^+(x-y) + \frac{1}{C_+}\sum \limits_{x \in \eta^+}\sum \limits_{y \in \eta^-}b^+(x-y).
\end{align*}
and hence by \eqref{MCSL:72} and \eqref{MCSL:73}
\begin{align*}
 c_+(\eta) &\leq \left( m^+ + C_-\| b^- \|_{L^1} + C_+\| a^- \|_{L^1} + \| a^+ \|_{L^1} + \frac{C_-}{C_+}\| b^+ \|_{L^1} + \frac{b}{C_+}\right)|\eta^+|
 \\ &\ \ \ + \left(\frac{\theta}{C_+} + 1\right)\sum \limits_{x \in \eta^+}\sum \limits_{y \in \eta^+ \backslash x}a^-(x-y) + \left( \frac{\vartheta}{C_+} + 1\right)\sum \limits_{x \in \eta^+}\sum \limits_{y \in \eta^-}b^-(x-y)
 \\ &\leq (1+a_+) M_+(\eta),
\end{align*}
where $M_+(\eta) = m^+|\eta^+| + \sum_{x \in \eta^+}\sum_{y \in \eta^+ \backslash x}a^+(x-y) + \sum_{x \in \eta^+}\sum_{y \in \eta^-}b^-(x-y)$ and
\[
 a_+ = \max\left\{ \frac{C_-\| b^- \|_{L^1} + C_+\| a^- \|_{L^1} + \| a^+ \|_{L^1} + \frac{C_-}{C_+}\| b^+ \|_{L^1} + \frac{b}{C_+}}{m^+}, \frac{\theta}{C_+},\frac{\vartheta}{C_+} \right\}.
\]
Using \eqref{MCSL:75} we conclude that $a_+ \in (0,1)$ which proves (S2).

The unique invariant measure for the environment $\mu_{\mathrm{inv}}$ is the Gibbs measure with activity $z^-$ and potential $\psi$.
The averaged rates are given by
\begin{align*}
 \overline{d}(x,\gamma^+) &= m^+ + \overline{m}(x) + \sum \limits_{y \in \gamma^+ \backslash x}a^-(x-y),
 \\ \overline{b}(x,\gamma^+) &= \overline{\lambda}(x) + \sum \limits_{y \in \gamma^+}a^+(x-y),
\end{align*}
where $\overline{\lambda}(x) := \int_{\Gamma^-} \sum_{y \in \gamma^-}b^+(x-y)d \mu_{\mathrm{inv}}(\gamma^-)$ and
$\overline{m}(x) := \int_{\Gamma^-} \sum_{y \in \gamma^-}b^-(x-y)d \mu_{\mathrm{inv}}(\gamma^-)$.
It remains to show that (AV1) -- (AV3) are satisfied. First observe that (AV1) holds with
\begin{align*}
 \overline{D}(x,\eta^+) &= 0^{|\eta^+|}(m^+ + \overline{m}(x)) + \1_{\Gamma_1^+}(\eta^+)\sum \limits_{y \in \eta^+}a^-(x-y),
 \\ \overline{B}(x,\eta^+) &= \overline{\lambda}(x)0^{|\eta^+|} + \1_{\Gamma_0^+}(\eta^+)\sum \limits_{y \in \eta^+}a^+(x-y).
\end{align*}
Concerning condition (AV3) we first observe that
\[
 \overline{\lambda}(x) = \int \limits_{\R^d}b^+(x-y)k_{\mathrm{inv}}(\{y\})dy \leq C_-\| k_{\mathrm{inv}} \|_{\K_{C_-}}\| b^+ \|_{L^1}
\]
from which we obtain $\overline{b}(x,\eta^+) \leq |\eta^+|\| a^+\|_{L^{\infty}} + C_-\| k_{\mathrm{inv}} \|_{\K_{C_-}}\| b^+ \|_{L^1}$.
For the last condition we obtain similarly to (S2)
\begin{align*}
 \overline{c}(\eta^+) &= \sum \limits_{x \in \eta^+}\left(m^+ + C_+\| a^- \|_{L^1} + \| a^+ \|_{L^1} + \frac{b}{C_+}\right)
 \\ &\ \ \ + \sum \limits_{x \in \eta^+}\left( \overline{m}(x) + \frac{\overline{\lambda}(x)}{C_+}\right) + \left( 1 + \frac{\theta}{C_+}\right)\sum \limits_{x \in \eta^+}\sum \limits_{y \in \eta^+ \backslash x}a^-(x-y)
 \\ &\leq (1+\overline{a}) \sum \limits_{x \in \eta^+}\left( m^+ + \overline{m}(x)\right),
\end{align*}
where we have used \eqref{MCSL:73} to obtain $\overline{\lambda}(x) \leq \vartheta \overline{m}(x)$ and by \eqref{MCSL:75}
\[
 \overline{a} := \max\left\{ \frac{C_+\| a^- \|_{L^1} + \| a^+ \|_{L^1} + \frac{b}{C_+}}{m^+}, \frac{\vartheta}{C_+}, \frac{\theta}{C_+}\right\} \leq a_+ \in (0,1).
\]
Under the given conditions we can apply Theorem \ref{TH:ERGODICITY} for $\overline{L}$ instead of $L^-$.
Let $\overline{\mu}_{\mathrm{inv}}$ be the corresponding unique invariant measure. 
Without interactions with the environment, i.e. $\overline{m} = \overline{\lambda} = 0$, we have $\overline{\mu}_{\mathrm{inv}} = \delta_{\emptyset}$.
In the presence of interactions, however, the invariant measure is non-degenerated, i.e. $\overline{\mu}_{\mathrm{inv}} \neq \delta_{\emptyset}$.

\subsection{Density dependent branching in Glauber environment}
Suppose that the environment is, as before, a Glauber dynamics with parameters $z^-, \psi$ (see \eqref{GLAUBER:ENV}).
For the system we assume that
\begin{align*}
 d^+(x,\gamma) &= m^+\exp\left( E_{\kappa}(x,\gamma^+)\right),
 \\ b^+(x, \gamma) &= \sum \limits_{y \in \gamma^+}\exp\left(-E_{\phi}(y, \gamma^-)\right) a^+(x-y).
\end{align*}
This model describes a branching process with strong (exponential) killing rate where the branching rate, in addition,
can be slowed down by interaction with the environment. This model is a prototype of a branching process inside a ''delirious'' environment.
We make the following assumptions on the parameters of the system:
\begin{enumerate}
 \item[(a)] $z^- > 0$, $\psi \geq 0$ is symmetric with $\beta(\psi) < \infty$ and there exists $C_- > 0$ such that \eqref{GLAUBER:ENV:COND} holds.
 \item[(b)] $m^+ > 0$ and $a^+, \phi, \kappa \geq 0$ are symmetric, with $\beta(\phi) < \infty$ and $a^+$ is integrable.
 Moreover $\kappa, a^+$ are bounded.
 \item[(c)] There exist constants $\vartheta > 0$ and $b \geq 0$ such that for all $\eta^+ \in \Gamma_0$
 \begin{align*}
  \sum \limits_{x \in \eta^+}\sum \limits_{y \in \eta^+ \backslash x}a^+(x-y) \leq \vartheta \sum \limits_{x \in \eta^+}\sum \limits_{y \in \eta^+ \backslash x}\kappa(x-y) + b|\eta^+|.
 \end{align*}
 Finally we have $\beta(-\kappa) < \infty$ and there exists $C_+ > 0$ such that
 \begin{align*}
   e^{C_+\beta(-\kappa)} + \frac{ e^{C_-\beta(\phi)}}{m^+ C_+}\max\{ C_+\| a^+ \|_{L^1} + b, \vartheta\} < 2.
 \end{align*}
\end{enumerate}
Then conditions (E1) -- (E3), (S1) -- (S3) and (AV1) -- (AV3) are satisfied.
Conditions (E1) -- (E3) have been shown in the first example. As before one can show that (S3) holds and (S1) is satisfied where
\begin{align*}
 D^+(x,\eta) &= 0^{|\eta^-|}m^+ \prod \limits_{y \in \eta^+}\left( e^{\kappa(x-y)} - 1 \right),
 \\ B^+(x,\eta) &= \1_{\Gamma_1^+}(\eta^+)\sum \limits_{y \in \eta^+}\prod \limits_{z \in \eta^-}\left( e^{-\phi(x-z)} - 1 \right)a(x-y).
\end{align*}
Then a computation shows that 
\begin{align*}
 c_+(\eta) &\leq m^+ e^{C_+ \beta(-\kappa)}\sum \limits_{x \in \eta^+}e^{E_{\kappa}(x,\eta^+ \backslash x)}
 + \frac{e^{C_- \beta(\phi)}}{C_+}\sum \limits_{x \in \eta^+}\sum \limits_{y \in \eta^+ \backslash x}a^+(x-y) + \frac{C_+}{C_-}\| a^+ \|_{L^1} e^{C_- \beta(\phi)} |\eta^+|
 \\ &\leq m^+ e^{C_+ \beta(-\kappa)}\sum \limits_{x \in \eta^+}e^{E_{\kappa}(x,\eta^+ \backslash x)}
 + \vartheta \frac{e^{C_- \beta(\phi)}}{C_+}\sum \limits_{x \in \eta^+}\sum \limits_{y \in \eta^+ \backslash x}\kappa(x-y)
 \\ &\ \ \ + \left(\frac{e^{C_- \beta(\phi)}}{C_+}b +\frac{C_+}{C_-}\| a^+ \|_{L^1} e^{C_- \beta(\phi)} \right)|\eta^+|
 \\ &\leq \left( m^+ e^{C_+ \beta(-\kappa)} + \max\{ b + C_+\|a^+\|_{L^1}, \vartheta\}\frac{e^{C_- \beta(\phi)}}{C_+} \right)\sum \limits_{x \in \eta^+}e^{E_{\kappa}(x,\eta^+ \backslash x)},
\end{align*}
where we have used
\[
 \sum \limits_{x \in \eta^+}e^{E_{\kappa}(x,\eta^+ \backslash x)} \geq |\eta^+| + \sum \limits_{x \in \eta^+}\sum \limits_{y \in \eta^+ \backslash x}\kappa(x-y).
\]
This shows (S2). Let us show (AV1) -- (AV3). The averaged birth-and-death rates are given by
\begin{align*}
 \overline{d}(x,\gamma^+) &= m^+\exp\left( E_{\kappa}(x,\gamma^+ \backslash x)\right),
 \\ \overline{b}(x,\gamma^+) &= \sum \limits_{y \in \gamma^+} \overline{\lambda}(y) a^+(x-y)
\end{align*}
with $\overline{\lambda}(y) := \int_{\Gamma^-}\exp\left(-E_{\phi^-}(y, \gamma^-)\right)d \mu_{\mathrm{inv}}(\gamma^-) \leq 1$.
As before (AV1) and (AV3) hold with
\begin{align*}
 \overline{D}(x,\eta^+) &= m^+ \prod \limits_{y \in \eta^+}\left( e^{\kappa(x-y)} - 1\right),
 \\ \overline{B}(x,\eta^+) &= \1_{\Gamma_1^+}(\eta^+) \sum \limits_{y \in \eta^+}\overline{\lambda}(y)a^+(x-y).
\end{align*}
Then for $\xi^+ \cap \zeta^+ = \emptyset$ we get
\begin{align*}
 \overline{B}(x,\zeta^+ \cup \xi^+) &= 0^{|\xi^+|}\1_{\Gamma_1^+}(\zeta^+) \sum \limits_{y \in \zeta^+}\overline{\lambda}(y)a^+(x-y)
 + 0^{|\zeta^+|}\1_{\Gamma_1^+}(\xi^+)\sum \limits_{y \in \xi^+}\overline{\lambda}(y)a^+(x-y)
\end{align*}
and hence
\begin{align*}
 \left| \sum \limits_{\zeta^+ \subset \eta^+ \backslash x}\overline{B}(x,\zeta^+ \cup \xi^+)\right|
 \leq 0^{|\xi^+|}\sum \limits_{y \in \eta^+ \backslash x}\overline{\lambda}(y)a^+(x-y)
 + \1_{\Gamma_1^+}(\xi^+)\sum \limits_{y \in \xi^+}\overline{\lambda}(y)a^+(x-y).
\end{align*}
A similar computation for $\overline{D}$ gives, recall $\overline{\lambda} \leq 1$,
\begin{align*}
 \overline{c}(\eta^+) &\leq m^+ e^{C_+ \beta(-\kappa)}\sum \limits_{x \in \eta^+}e^{E_{\kappa}(x,\eta^+ \backslash x)}
 + \sum \limits_{x \in \eta^+}\int \limits_{\R^d}\overline{\lambda}(y)a^+(x-y)dy
 + \frac{1}{C_+} \sum \limits_{x \in \eta^+}\sum \limits_{y \in \eta^+ \backslash x}\overline{\lambda}(y)a^+(x-y)
 \\ &\leq m^+ e^{C_+ \beta(-\kappa)}\sum \limits_{x \in \eta^+}e^{E_{\kappa}(x,\eta^+ \backslash x)}
 + \left( \| a^+ \|_{L^1} + \frac{b}{C_+} \right)|\eta^+| + \frac{\vartheta}{C_+}\sum \limits_{x \in \eta^+}\sum \limits_{y \in \eta^+ \backslash x}\kappa(x-y)
 \\ &\leq \left( m^+ e^{C_+ \beta(-\kappa)} + \max\left\{ \| a^+ \|_{L^1} + \frac{b}{C_+}, \frac{\vartheta}{C_+} \right\}\right)\sum \limits_{x \in \eta^+}e^{E_{\kappa}(x,\eta^+ \backslash x)}.
\end{align*}
It follows from
\[
 m^+ e^{C_+ \beta(-\kappa)} + \max\left\{ \| a^+ \|_{L^1} + \frac{b}{C_+}, \frac{\vartheta}{C_+} \right\}
 \leq m^+ e^{C_+ \beta(-\kappa)} + \frac{e^{C_- \beta(\phi)}}{C_+}\max\left\{ C_+\| a^+ \|_{L^1} + b,\vartheta\right\}
\]
that (AV2) is satisfied.

\subsection{Two interacting BDLP-models}
Suppose that the environment is given by an BDLP model with immigration parameter $z > 0$, i.e.
\begin{align*}
 d^-(x,\gamma^-) &= m^- + \sum \limits_{y \in \gamma^-} a^-(x-y),
 \\  b^-(x,\gamma^-) &= \sum \limits_{y \in \gamma^-}a^+(x-y) + z.
\end{align*}
For the system we suppose that it is also an BDLP model with additional killing and branching caused by the environment at additive rates, i.e.
\begin{align*}
 d^+(x,\gamma) &= m^+ + \sum \limits_{y \in \gamma^+ } b^-(x-y) + \sum \limits_{y \in \gamma^-}\varphi^-(x-y),
 \\ b^+(x,\gamma) &= \sum \limits_{y \in \gamma^+}b^+(x-y) + \sum \limits_{y \in \gamma^-}\varphi^+(x-y).
\end{align*}
We make the following assumptions on the parameters of the model
\begin{enumerate}
 \item[(a)] $z,m^+, m^- > 0$ and $a^{\pm}, b^{\pm}, \varphi^{\pm}$ are non-negative, symmetric, integrable and bounded.
 \item[(b)] There exist constants $b_1, b_2 \geq 0$ and $\vartheta_1, \vartheta_2, \vartheta_3 > 0$ such that
 \begin{align*}
  \sum \limits_{x \in \eta^+}\sum \limits_{y \in \eta^+ \backslash x}b^+(x-y) &\leq \vartheta_1 \sum \limits_{x \in \eta^+}\sum \limits_{y \in \eta^+ \backslash x}b^-(x-y) + b_1|\eta^+|
  \\ \sum \limits_{x \in \eta^-}\sum \limits_{y \in \eta^- \backslash x}a^+(x-y) &\leq \vartheta_2 \sum \limits_{x \in \eta^-}\sum \limits_{y \in \eta^- \backslash x}a^-(x-y) + b_2|\eta^-|,
 \end{align*}
 and $\varphi^+ \leq \vartheta_3 \varphi^-$ hold. 
 Finally there exist $C_+ > \vartheta_1, \vartheta_3$ and $C_- > \vartheta_2$ such that
 \begin{align*}
   m^+ &> C_+\| b^- \|_{L^1} + C_-\| \varphi^- \|_{L^1} + \frac{b_1}{C_+} + \| b^+\|_{L^1} + \| \varphi^+ \|_{L^1},
  \\ m^- &> C_-\| a^- \|_{L^1} + \frac{b_2 + z}{C_-} + \| a^+\|_{L^1}.
 \end{align*}
\end{enumerate}
Then conditions (E1) -- (E3), (S1) -- (S3) and (AV1) -- (AV3) are satisfied. First it is clear that (E1), (S1) and (E3), (S3) are satisfied with
\begin{align*}
 D^-(x,\eta^-) &= m^- 0^{|\eta^-|} + \1_{\Gamma_1^-}(\eta^-)\sum \limits_{y \in \eta^-}a^-(x-y),
 \\ B^-(x,\eta^-) &= z 0^{|\eta^-|} + \1_{\Gamma_1^-}(\eta^-)\sum \limits_{y \in \eta^-}a^+(x-y),
 \\ D^+(x,\eta) &= m^+ 0^{|\eta|} + 0^{|\eta^-|}\1_{\Gamma_1^+}(\eta^+)\sum \limits_{y \in \eta^+}b^-(x-y) + 0^{|\eta^+|}\1_{\Gamma_1^-}(\eta^-)\sum \limits_{y \in \eta^-}\varphi^-(x-y),
 \\ B^+(x,\eta) &= 0^{|\eta^-|}\1_{\Gamma_1^+}(\eta^+)\sum \limits_{y \in \eta^+}b^+(x-y) + 0^{|\eta^+|}\1_{\Gamma_1^-}(\eta^-)\sum \limits_{y \in \eta^-}\varphi^+(x-y).
\end{align*}
Hence we obtain after some computations
\begin{align*}
 c_-(\eta^-) &\leq \left( m^- + C_- \| a^- \|_{L^1} + \frac{z}{C_-} + \| a^+ \|_{L^1}\right)|\eta^-| 
 + \sum \limits_{x \in \eta^-}\sum \limits_{y \in \eta^- \backslash x}a^-(x-y)
 + \frac{1}{C_-}\sum \limits_{x \in \eta^-}\sum \limits_{y \in \eta^- \backslash x}a^+(x-y)
 \\ &\leq \left( m^- + C_- \| a^- \|_{L^1} + \frac{z}{C_-} + \| a^+ \|_{L^1} + \frac{b_2}{C_-}\right)|\eta^-| 
 + \left( 1 + \frac{\vartheta_2}{C_-}\right)\sum \limits_{x \in \eta^-}\sum \limits_{y \in \eta^- \backslash x}a^-(x-y)
\end{align*}
and likewise
\begin{align*}
 c_+(\eta) &\leq \left( m^+ + C_+ \| b^- \|_{L^1} + C_- \| \varphi^- \|_{L^1} + \| b^+ \|_{L^1} + \frac{C_-}{C_+}\| \varphi^+ \|_{L^1}\right)|\eta^+|
 \\ &\ \ \ + \sum \limits_{x \in \eta^+}\sum \limits_{y \in \eta^+ \backslash x}b^-(x-y)
 + \frac{1}{C_+}\sum \limits_{x \in \eta^+}\sum \limits_{y \in \eta^+ \backslash x}b^+(x-y)
 \\ &\ \ \ + \sum \limits_{x \in \eta^+}\sum \limits_{y \in \eta^-}\varphi^-(x-y)
 + \frac{1}{C_+}\sum \limits_{x \in \eta^+}\sum \limits_{y \in \eta^-}\varphi^+(x-y)
 \\ &\leq \left( m^+ + C_+ \| b^- \|_{L^1} + C_- \| \varphi^- \|_{L^1} + \| b^+ \|_{L^1} + \frac{C_-}{C_+}\| \varphi^+ \|_{L^1} + \frac{b_1}{C_+}\right)|\eta^+|
 \\ &\ \ \ + \left( 1 + \frac{\varphi_1}{C_+}\right)\sum \limits_{x \in \eta^+}\sum \limits_{y \in \eta^+ \backslash x}b^-(x-y)
 + \left( 1 + \frac{\varphi_3}{C_+}\right)\sum \limits_{x \in \eta^+}\sum \limits_{y \in \eta^-}\varphi^-(x-y).
\end{align*}
In view of the assumptions made on the parameters it is easily seen that (E2) and (S2) hold.
Let us show that (AV1) -- (AV3) hold. The averaged birth-and-death rates are given by
\begin{align*}
 \overline{d}(x,\gamma^+) &= m^+  + \overline{\varphi}^-(x) + \sum \limits_{y \in \gamma^+ \backslash x}b^-(x-y),
 \\ \overline{b}(x,\gamma^+) &= \sum \limits_{y \in \gamma^+}b^+(x-y) + \overline{\varphi}^+(x),
\end{align*}
where $\overline{\varphi}^{\pm}(x) = \int_{\Gamma^-}\sum_{y \in \gamma^-}\varphi^{\pm}(x-y)d \mu_{\mathrm{inv}}(\gamma^-)$.
Clearly (AV1) and (AV3) hold with
\begin{align*}
 \overline{D}(x,\eta^+) &= \left( m^+ + \overline{\varphi}^-(x)\right)0^{|\eta^+|} + \1_{\Gamma_1^+}(\eta^+)\sum \limits_{y \in \eta^+}b^-(x-y),
 \\ \overline{B}(x,\eta^+) &= \overline{\varphi}^+(x) 0^{|\eta^+|} + \1_{\Gamma_1^+}(\eta^+)\sum \limits_{y \in \eta^+}b^+(x-y).
\end{align*}
Then, as before, we get by assumption (c)
\begin{align*}
 \overline{c}(\eta^+) &\leq \left( m^+ + \overline{\varphi}^-(x) + C_+ \| b^- \|_{L^1} + \frac{\overline{\varphi}^+(x)}{C_+} + \| b^+ \|_{L^1}\right)|\eta^+|
 \\ &\ \ \ + \sum \limits_{x \in \eta^+}\sum \limits_{y \in \eta^+ \backslash x}b^-(x-y)
 + \frac{1}{C_+}\sum \limits_{x \in \eta^+}\sum \limits_{y \in \eta^+ \backslash x}b^+(x-y)
 \\ &\leq \left( m^+ + \overline{\varphi}^-(x) + C_+ \| b^- \|_{L^1} + \frac{\overline{\varphi}^+(x)}{C_+} + \| b^+ \|_{L^1} + \frac{b_1}{C_+}\right)|\eta^+|
 + \left( 1 + \frac{\vartheta_1}{C_+}\right)\sum \limits_{x \in \eta^+}\sum \limits_{y \in \eta^+ \backslash x}b^-(x-y)
\end{align*}
which shows condition (AV2).

\section{Proof of Theorem \ref{TH:FPE} and Theorem \ref{TH:ERGODICITY}}

\subsection{Proof of Theorem \ref{TH:FPE}}
In this section we closely follow the arguments in \cite{FK16b} (see also \cite{F16WR}).
Since the necessary computations are very similar to the latter works, we give only the main steps of proof.
Let $\Lb_{C_-}$ be the Banach space of integrable functions with norm
\[
 \| G \|_{\Lb_{C_-}} = \int \limits_{\Gamma_0^-}|G(\eta^-)|C_-^{|\eta^-|}d\lambda(\eta^-).
\]
Define an operator $\widehat{L}^-$ on the domain $D(\widehat{L}^-) := \{ G \in \Lb_{C_-} \ | \ M_- \cdot G \in \Lb_{C_-} \}$ by
\begin{align*}
 (\widehat{L}^-G)(\eta^-) &= - \sum \limits_{\xi^- \subset \eta^-}G(\xi^-) \sum \limits_{x \in \xi^-}\sum \limits_{\zeta^- \subset \xi^- \backslash x}D^-(x,\eta^- \backslash \xi^- \cup \zeta^-)
 \\ &+ \sum \limits_{\xi^- \subset \eta^-}\int \limits_{\R^d}G(\xi^- \cup x)\sum \limits_{\zeta^- \subset \xi^-}B^-(x,\eta^- \backslash \xi^- \cup \zeta^-)d x.
\end{align*}
We will see that this operator is related to $L^-$ via the $K$-transform.
\begin{proposition}\label{PROP:03}
 Suppose that (E1), (E2) are satisfied. Then $(\widehat{L}^-, D(\widehat{L}^-))$ is the generator of an analytic semigroup of contractions 
 $(T^-(t))_{t \geq 0}$ on $\Lb_{C_-}$. Moreover $B_{bs}(\Gamma_0^-)$ is a core for $(\widehat{L}^-, D(\widehat{L}^-))$.
\end{proposition}
\begin{proof}
 Consider the decomposition $\widehat{L}^- = A + B$ where $AG(\eta^-) = - M_-(\eta)G(\eta^-)$ and
 \begin{align*}
  (BG)(\eta^-) &= - \sum \limits_{\xi^- \subsetneq \eta^-}G(\xi^-) \sum \limits_{x \in \xi^-}\sum \limits_{\zeta^- \subset \xi^- \backslash x}D^-(x,\eta^- \backslash \xi^- \cup \zeta^-)
  \\ &+ \sum \limits_{\xi^- \subset \eta^-}\int \limits_{\R^d}G(\xi^- \cup x)\sum \limits_{\zeta^- \subset \xi^-}B^-(x,\eta^- \backslash \xi^- \cup \zeta^-)d x.
 \end{align*}
 Observe that $(A, D(\widehat{L}^-))$ is the generator of a positive, analytic semigroup of contractions on $\Lb_{C_-}$.
 Define another positive operator on $D(\widehat{L}^-)$ by
 \begin{align*}
  (B_*G)(\eta^-) &= - \sum \limits_{\xi^- \subsetneq \eta^-}G(\xi^-) \sum \limits_{x \in \xi^-}\left|\sum \limits_{\zeta^- \subset \xi^- \backslash x}D^-(x,\eta^- \backslash \xi^- \cup \zeta^-)\right|
 \\ &+ \sum \limits_{\xi^- \subset \eta^-}\int \limits_{\R^d}G(\xi^- \cup x)\left|\sum \limits_{\zeta^- \subset \xi^-}B^-(x,\eta^- \backslash \xi^- \cup \zeta^-)\right|d x.
 \end{align*}
 By (E2) we can find $r \in (0,1)$ such that $a_- < 1+r < 2$. 
 Then, by \eqref{IBP}, a short computation shows that for $0 \leq G \in D(\widehat{L}^-)$ we have
 \begin{align*}
  \int \limits_{\Gamma_0^-} B_*G(\eta^-)C_-^{|\eta^-|}d\lambda(\eta^-) 
  \leq \int \limits_{\Gamma_0^-}\left( c_-(\eta^-) - M_-(\eta^-)\right)|G(\eta^-)|C_-^{|\eta^-|}d\lambda(\eta^-)
  \leq (a_- - 1)\| AG \|_{\Lb_{C_-}}
 \end{align*}
 which yields $\int_{\Gamma_0^-}( A + \frac{1}{r}B_*)G(\eta^-)C_-^{|\eta^-|}d\lambda(\eta^-) \leq 0$.
 Hence by \cite[Theorem 2.2]{TV06} $(A + B_*, D(\widehat{L}^-))$ is the generator of a positive semigroup $V(t)$ of contractions on $\Lb_{C_-}$.
 Applying \cite[Theorem 1.1]{AR91} together with $|BG| \leq B_*|G|$ it follows that $(A + B, D(\widehat{L}^-))$
 is the generator of an analytic semigroup $T^-(t)$ on $\Lb_{C_-}$. By \cite[Theorem 1.2]{AR91} we get $|T^-(t)G| \leq V(t)|G|$
 and since $V(t)$ is a semigroup of contractions, the same holds true for $T(t)$.
 For the last assertion let $G \in D(\widehat{L}^-)$ and set $G_n(\eta^-) = \1_{|\eta^-| \leq n}\1_{\eta^- \subset B_n}G(\eta^-)$
 where $B_n$ denotes the ball of diameter $n$ around zero. 
 Then it is easily seen that $G_n \longrightarrow G$ and $\widehat{L}^-G_n \longrightarrow \widehat{L}^-G$ in $\Lb_{C_-}$.
 \qed
\end{proof}
Let us now prove Lemma \ref{ENV:LEMMA00}.
\begin{proof}(Lemma \ref{ENV:LEMMA00})
 (a) Fix $\mu \in \mathcal{P}_{C_-}$. First note that the K-transform has an unique extension to a bounded linear operator
 $K: L^1(\Gamma_0^-, k_{\mu}d\lambda) \longrightarrow L^1(\Gamma^-, d\mu)$ such that \eqref{EQ:12} is absolutely convergent for $G \in L^1(\Gamma_0^-, k_{\mu}d\lambda)$
 and $\mu$-a.a. $\gamma^-$ (see \cite{KK02}). The assertion now follows from the following estimates
 \begin{align*}
  \int \limits_{\Gamma^-}\left( d^-(x,\gamma^-) + b^-(x,\gamma^-)\right)d\mu(\gamma^-)
  &= \int \limits_{\Gamma_0^-}\left( D^-(x,\eta^-) + B^-(x,\eta^-)\right)k_{\mu}(\eta^-)d\lambda(\eta^-)
  \\ &\leq \int \limits_{\Gamma_0^-}\left( |D^-(x,\eta^-)| + |B^-(x,\eta^-)|\right)k_{\mu}(\eta^-)d\lambda(\eta^-)
  \\ &\leq \| k_{\mu} \|_{\K_{C_-}} \max\{1, C_-\}c_-(\{x\}) \leq \| k_{\mu} \|_{\K_{C_-}} \max\{1, C_-\}a_- d^-(x,\emptyset).
 \end{align*}
 (b) Since for $\mu \in \mathcal{P}_{C_-}$ and $G \in B_{bs}(\Gamma_0^-)$ we have 
 $\widehat{L}G \in \Lb_{C_-} \subset L^1(\Gamma_0^-, k_{\mu}d\lambda)$ it follows that $K\widehat{L}G \in L^1(\Gamma^-,d\mu)$.
 Let $(K^{-1}F)(\eta^-) = \sum_{\xi^- \subset \eta^-}(-1)^{|\eta^- \backslash \xi^-|}F(\xi^-), \eta^- \in \Gamma_0^-$
 be the inverse transformation to \eqref{EQ:12}. Using the properties of $K^{-1}$ together with (E1) we obtain
 \begin{align*}
  (K^{-1}d^-(x,\cdot \cup \xi^- \backslash x))(\eta^- \backslash \xi^-) &= \sum \limits_{\zeta^- \subset \eta^- \backslash \xi^-}(-1)^{|(\eta^- \backslash \xi^-) \backslash \zeta^-|}\sum \limits_{\alpha \subset \zeta^- \cup \xi^- \backslash x}D^-(x,\alpha)
  \\ &= \sum \limits_{\zeta^- \subset \eta^- \backslash \xi^-}(-1)^{|(\eta^- \backslash \xi^-) \backslash \zeta^-|}\sum \limits_{\alpha_1 \subset \zeta^-}\sum \limits_{\alpha_2 \subset \xi^- \backslash x}D^-(x,\alpha_1 \cup \alpha_2)
  \\ &= \sum \limits_{\alpha_2 \subset \xi^- \backslash x}D^-(x,\eta^- \backslash \xi^- \cup \alpha_2)
 \end{align*}
 and similarly $(K^{-1}b^-(x,\cdot \cup \xi^-))(\eta^- \backslash \xi^-) = \sum_{\alpha_2 \subset \xi^-}B^-(x,\eta^- \backslash \xi^- \cup \alpha_2)$.
 Hence we have shown that 
 \begin{align*}
  (\widehat{L}^-G)(\eta^-) &= - \sum \limits_{\xi^- \subset \eta^-}G(\xi^-) \sum \limits_{x \in \xi^-}(K^{-1}d^-(x,\cdot \cup \xi^- \backslash x))(\eta^- \backslash \xi^-)
 \\ &+ \sum \limits_{\xi^- \subset \eta^-}\int \limits_{\R^d}G(\xi^- \cup x)(K^{-1}b(x,\cdot \cup \xi^-))(\eta^- \backslash \xi^-)d x.
 \end{align*}
 Now we may deduce from \cite[Proposition 3.1]{FKK12} that $K\widehat{L}^- = L^-KG$ for $G \in B_{bs}(\Gamma_0^-)$.
 \qed
\end{proof}
Let $(\Lb_{C_-})^*$ be the dual Banach space to $\Lb_{C_-}$. 
Using the duality $\langle G, k \rangle := \int_{\Gamma_0^-}G(\eta^-)k(\eta^-)d\lambda(\eta^-)$ it can be identified with $\K_{C_-}$. 
For $k \in \K_{C_-}$ let 
\begin{align*}
 (L^{\Delta,-}k)(\eta^-) &= - \sum \limits_{x \in \eta^-}\int \limits_{\Gamma_0^-}k(\eta^- \cup \xi^-)\sum \limits_{\zeta^- \subset \eta^- \backslash x}D^-(x,\zeta^- \cup \xi^-)d\lambda(\xi^-)
 \\ &\ \ \ + \sum \limits_{x \in \eta^-}\int \limits_{\Gamma_0^-}k(\eta^- \cup \xi^- \backslash x)\sum \limits_{\zeta^- \subset \eta^- \backslash x}B^-(x,\zeta^- \cup \xi^-)d\lambda(\xi^-).
\end{align*}
Note that $L^{\Delta,-}k$ is $\lambda$-a.e. well-defined, satisfies for any $k \in \K_{C_-}$ 
\begin{align*}
 |L^{\Delta,-}k(\eta^-)| \leq \| k \|_{\K_{C_-}}C_-^{|\eta^-|}c_-(\eta^-), \ \ \eta^- \in \Gamma_0^-,
\end{align*}
but, in general, $L^{\Delta,^-}k \not \in \K_{C_-}$.
\begin{lemma}\label{ENV:LEMMA01}
 Suppose that conditions (E1) and (E2) are satisfied. 
 Let $(\widehat{L}^{*,-}, D(\widehat{L}^{*,-}))$ be the adjoint operator to $(\widehat{L}^-, D(\widehat{L}^-))$ on $\K_{C_-}$.
 Then $L^{\Delta,-}k = \widehat{L}^{*,-}k$ for any $k \in D(\widehat{L}^{*,-})$ and
 \[
  D(\widehat{L}^{*,-}) = \{ k \in \K_{C_-} \ | \ L^{\Delta,-}k \in \K_{C_-} \}.
 \]
\end{lemma}
\begin{proof}
 Arguing similarly to \cite[Lemma 3.5]{FK16b} one can show that 
 \[
  \langle \widehat{L}^-G, k \rangle = \langle G, L^{\Delta,-}k \rangle, \ \ G \in D(\widehat{L}^-), \ \ k \in \K_{C_-}
 \]
 from which one can readily deduce the assertion.
 \qed
\end{proof}
The Cauchy problem
\begin{align}\label{EQ:01}
 \frac{d}{dt}k_t^- = L^{\Delta,^-}k_t^-, \ \ k_t^-|_{t=0} = k_0^-
\end{align}
is a Markov analogue of the BBGKY-hierarchy and describes the evolution of correlation functions corresponding to the Fokker-Planck equation.
Denote by $\widehat{T}^-(t)^*$ the adjoint semigroup to $T^-(t)$.
The next proposition gives existence and uniqueness of weak solutions to this hierarchy.
\begin{proposition}\label{PROP:04}
 Suppose that (E1) and (E2) are satisfied.
 \begin{enumerate}
  \item[(a)] For any $k_0^- \in \K_{C_-}$ the function $k_t^- = T^-(t)^*k_0^-$ satisfies
  \begin{align}\label{CORR}
    \langle G, k_t^-\rangle = \langle G, k_0^- \rangle + \int \limits_{0}^{t}\langle \widehat{L}^-G, k_s^- \rangle ds, \ \ t \geq 0
  \end{align}
   for any $G \in B_{bs}(\Gamma_0^-)$.
  \item[(b)] Let $(r_t)_{t \geq 0} \subset \K_{C_-}$ with $r_0 = k_0^-$ satisfy \eqref{CORR} and
   suppose that 
   \begin{align}\label{EQ:18}
    \sup \limits_{t \in [0,T]}\| r_t \|_{\K_{C_-}} < \infty, \ \ \forall T > 0.
   \end{align}
   Then $r_t = T^-(t)^*k_0$.
 \end{enumerate}
\end{proposition}
\begin{proof}
 Note that assertion (a) readily follows from the properties of the adjoint semigroup $T^-(t)^*$.
 Let us prove (b). It follows from \cite[Theorem 2.1]{WZ06} that there exists at most one solution to \eqref{CORR} 
 such that $t \longmapsto r_t$ is continuous w.r.t. the topology $\mathcal{C}$.
 Here $\mathcal{C}$ is the topology of uniform convergence on compact sets of $\Lb_{C_-}$
 given by a basis of neighbourhoods with
 \[
  \left \{ r \in \K_{C_-} \ : \ |\langle G, r \rangle - \langle G, k \rangle| < \e, \ \ \forall G \in K \right\}
 \]
 where $\e > 0$, $k \in \K_{C_-}$ and $K$ is a compact subset of $\Lb_{C_-}$.
 Using $\widehat{L}^-G \in \Lb_{C_-}$, \eqref{EQ:18} and \eqref{CORR} we see that $t \longmapsto \langle G, r_t \rangle$ 
 is continuous for any $G \in B_{bs}(\Gamma_0^-)$. Since $B_{bs}(\Gamma_0^-)$ is dense in $\Lb_{C_-}$, by \eqref{EQ:18} we can show that
 $t \longmapsto r_t$ is continuous w.r.t. $\sigma(\K_{C_-}, \Lb_{C_-})$.
 By \cite[Lemma 1.10]{WZ06} it follows that $t \longmapsto r_t$ is also continuous w.r.t. $\mathcal{C}$ which proves the assertion.
 \qed
\end{proof}
In the next step we prove the equivalence between solutions to \eqref{EQ:06} and \eqref{CORR}.
\begin{lemma}
 Suppose that conditions (E1) -- (E2) are satisfied.
 Let $(\mu_t)_{t \geq 0} \subset \mathcal{P}_{C_-}$, denote by $(k_t)_{t \geq 0}$ the corresponding correlation functions and assume that
 \[
    \sup \limits_{t \in [0,T]}\| k_t \|_{\K_{C_-}} < \infty, \ \ \forall T > 0.
 \]
 Then $(\mu_t)_{t \geq 0}$ satisfies \eqref{EQ:06} if and only if $(k_t)_{t \geq 0}$ satisfies \eqref{CORR}.
\end{lemma}
\begin{proof}
 Recall that $K: L^1(\Gamma_0, k_{\mu_t}d\lambda) \longrightarrow L^1(\Gamma^-,d\mu_t)$ is a bounded linear operator.
 Hence for $F = KG$ with $G \in B_{bs}(\Gamma_0^-)$ we get $L^-F = K\widehat{L}^-G \in L^1(\Gamma^-,d\mu_t)$.
 Moreover, by \eqref{EQ:14} we get
 \begin{align*}
  \int \limits_{\Gamma_0^-}G(\eta^-)k_t(\eta^-)d\lambda(\eta^-) &= \int \limits_{\Gamma^-}F(\gamma^-)d\mu_t(\gamma^-),
  \\ \int \limits_{\Gamma_0^-}(\widehat{L}^-G)(\eta^-)k_t(\eta^-)d\lambda(\eta^-) &= \int \limits_{\Gamma^-}(L^-F)(\gamma^-)d\mu_t(\gamma^-)
 \end{align*}
 which then implies the assertion.
 \qed
\end{proof}
Let $\mu_0^- \in \mathcal{P}_{C_-}$ with correlation function $k_0^-$ and let $k_t^-$ be the unique solution to \eqref{CORR}.
It remains to show that there exist $\mu_t^- \in \mathcal{P}_{C_-}$ such that $k_{\mu_t^-} = k_t^-$.
For this purpose we show that $k_t^-$ is positive definite in the sense of Lenard, i.e.
\begin{align}\label{EQ:19}
 \int \limits_{\Gamma_0^-}G(\eta^-)k_t^-(\eta^-)d\lambda(\eta^-) \geq 0, \ \ \forall G \in B_{bs}(\Gamma_0^-) \text{ with } KG \geq 0.
\end{align}
First we prove the following lemma.
\begin{lemma}
 Let $(u_t^{\delta})_{t \geq 0}$ be such that
 \[
  |u_t^{\delta}(\eta^-)| \leq A C_-^{|\eta^-|} \prod \limits_{x \in \eta^-}R_{\delta}(x), \ \ \eta \in \Gamma_0^-, \ \ t \geq 0
 \]
 and suppose that for any bounded measurable function $F: \Gamma_0^- \longrightarrow \R$
 \[
  \langle F, \mathcal{H}u_t^{\delta} \rangle = \langle F, \mathcal{H}u_0^{\delta} \rangle + \int \limits_{0}^{t}\langle L_{\delta}^-F, \mathcal{H}u_s^{\delta}\rangle ds, \ \ t \geq 0
 \]
 holds, where $L_{\delta}^-F$ is given by \eqref{MCSL:21} with $b^-$ replaced by $R_{\delta}(x)b^-$ and
 \[
  (\mathcal{H}u_s^{\delta})(\eta^-) = \int \limits_{\Gamma_0^-}(-1)^{|\xi^-|}u_{s}^{\delta}(\eta^- \cup \xi^-)d\lambda(\xi^-), \ \ \eta^- \in \Gamma_0^-.
 \]
 Then $t \longmapsto u_t \in L^1(\Gamma_0^-, d\lambda)$ is continuous w.r.t. the norm.
\end{lemma}
\begin{proof}
 Take $0 \leq s < t$ and let $F$ be bounded and measurable with $\| F\|_{L^{\infty}} \leq 1$. Then
 \begin{align*}
  &\ \left| \int \limits_{\Gamma_0^-}F(\eta^-)\mathcal{H}u_t^{\delta}(\eta^-)d\lambda(\eta^-) - \int \limits_{\Gamma_0^-}F(\eta^-)\mathcal{H}u_s^{\delta}(\eta^-)d\lambda(\eta^-) \right|
  \\ &\leq \int \limits_{s}^{t} \int \limits_{\Gamma_0^-}|(L_{\delta}^-F)(\eta^-)||\mathcal{H}u_r^{\delta}(\eta^-)|d\lambda(\eta^-)dr
  \leq 2 (t-s)c_0 \sup \limits_{r \in [s,t]} \int \limits_{\Gamma_0^-}e^{c_1|\eta^-|}|\mathcal{H}u_r^{\delta}(\eta^-)|d\lambda(\eta^-)
 \end{align*}
 for some constants $c_0,c_1 > 0$ where we have used (E1). Moreover we have by \eqref{IBP}
 \begin{align*}
  &\ \int \limits_{\Gamma_0^-}e^{c_1|\eta^-|}|\mathcal{H}u_r^{\delta}(\eta^-)|d\lambda(\eta^-)
  \leq \int \limits_{\Gamma_0^-}\int \limits_{\Gamma_0^-}e^{c_1 |\eta^-|}|u_r^{\delta}(\eta^- \cup \xi^-)|d\lambda(\eta^-)d\lambda(\xi^-)
  \\ &= \int \limits_{\Gamma_0^-}\sum \limits_{\eta^- \subset \xi^-}e^{c_1 |\eta^-|}|u_r^{\delta}(\xi^-)|d\lambda(\xi^-)
     \leq C \int \limits_{\Gamma_0^-}\left(1 + e^{c_1}\right)^{|\xi^-|} C_-^{|\xi^-|} \prod \limits_{x \in \xi^-}R_{\delta}(x) d\lambda(\xi^-)
  \\ &\leq C \sum \limits_{k=0}^{\infty}\frac{C_-^k (1 + e^{c_1})^k \| R_{\delta} \|_{L^1}^k}{k!} < \infty.
 \end{align*}
 Taking the supremum over all such $F$ gives $\| \mathcal{H}u_t^{\delta} - \mathcal{H}u_s^{\delta}\|_{L^1} \leq A(t-s)$ for some constant $A > 0$.
 \qed
\end{proof}
Property \eqref{EQ:19} can be shown by following the arguments given in \cite[Lemma 3.18]{FK16b} or \cite{F16WR}.
The only difference is that in \cite[Lemma 3.18, p. 363]{FK16b} a stronger condition then (E2) was used to prove the assertion of previous lemma.
Secondly it is worth to mention that in \cite[section 3]{FK16b} a more technical condition then (E3) was used. 
This condition, however, can be deduced from (E3) by applying \cite[Theorem 2.2]{TV06} together with \cite[Proposition 5.1]{TV06}.

\subsection{Proof of Theorem \ref{TH:ERGODICITY}}
Using \eqref{EQ:14} together with Lemma \ref{ENV:LEMMA01} we see that any invariant measure $\mu_{\mathrm{inv}} \in \mathcal{P}_{C_-}$ satisfies
\[
 0 = \int \limits_{\Gamma_0^-}(KG)(\gamma^-)d\mu_{\mathrm{inv}}(\gamma^-) 
 = \int \limits_{\Gamma_0^-}(\widehat{L}^-G)(\eta^-)k_{\mathrm{inv}}(\eta^-)d\lambda(\eta^-) 
 = \int \limits_{\Gamma_0^-}G(\eta^-)(L^{\Delta,-}k_{\mathrm{inv}})(\eta^-)d\lambda(\eta^-)
\]
for any $G \in B_{bs}(\Gamma_0^-)$ and hence $L^{\Delta,-}k_{\mathrm{inv}} = 0$. 
The next lemma states that this equation has, indeed, exactly one solution.
\begin{lemma}
 The equation
 \begin{align*}
  L^{\Delta,-}k_{\mathrm{inv}} = 0, \ \ k_{\mathrm{inv}}(\emptyset) = 1
 \end{align*}
 has a unique solution $k_{\mathrm{inv}} \in \K_{C_-}$. 
 Moreover, for any $\rho \in \R$ the function $k_{\rho}(\eta^-) = \rho 0^{|\eta^-|} + \rho k_{\mathrm{inv}}(\eta^-)$ is the unique solution to
 \begin{align}\label{EQ:21}
  L^{\Delta,-}k_{\rho} = 0, \ \ k_{\rho}(\emptyset) = \rho.
 \end{align}
\end{lemma}
The following proof is based on the ideas from \cite{FKK12}.
\begin{proof}
 Consider the decomposition $\K_{C_-} = \K_{C_-}^0 \oplus \K_{C_-}^{\geq 1}$ with
 \begin{align}\label{EQ:22}
  \K_{C_-}^0 = \{ k \in \K_{C_-} \ | \ k(\eta^-) = k(\emptyset) 0^{|\eta^-|} \}, \quad \K_{C_-}^{\geq 1} = \{ k \in \K_{C_-} \ | \ k(\emptyset) = 0 \}.
 \end{align}
 For $k \in \K_{C_-}$ write $k = 0^{|\eta^-|}\rho + \widetilde{k}$ with $\widetilde{k} \in \K_{C_-}^{\geq 1}$ and $\rho = k(\emptyset)$. Hence obtain
 \begin{align*}
  (L^{\Delta,-}k)(\eta^-) &= - M_-(\eta^-)\widetilde{k}(\eta^-) - \sum \limits_{x \in \eta^-}\int \limits_{\Gamma_0^- \backslash \emptyset}\widetilde{k}(\eta^- \cup \xi^-)\sum \limits_{\zeta^- \subset \eta^- \backslash x}D^-(x,\zeta^- \cup \xi^-)d\lambda(\xi^-)
 \\ &\ \ \ + \sum \limits_{x \in \eta^-}\int \limits_{\Gamma_0^-}\widetilde{k}(\eta^- \cup \xi^- \backslash x)\sum \limits_{\zeta^- \subset \eta^- \backslash x}B^-(x,\zeta^- \cup \xi^-)d\lambda(\xi^-)
 + \rho \1_{\Gamma_1^-}(\eta^-)\sum \limits_{x \in \eta^-}B^-(x,\emptyset).
 \end{align*}
 We let $S\widetilde{k}(\emptyset) = 0$ and for $\eta^- \neq \emptyset$
 \begin{align*}
  (S\widetilde{k})(\eta^-) &= - \frac{1}{M_-(\eta^-)}\sum \limits_{x \in \eta^-}\int \limits_{\Gamma_0^- \backslash \emptyset}\widetilde{k}(\eta^- \cup \xi^-)\sum \limits_{\zeta^- \subset \eta^- \backslash x}D^-(x,\zeta^- \cup \xi^-)d\lambda(\xi^-)
  \\ &\ \ \ + \frac{1}{M_-(\eta^-)}\sum \limits_{x \in \eta^-}\int \limits_{\Gamma_0^-}\widetilde{k}(\eta^- \cup \xi^- \backslash x)\sum \limits_{\zeta^- \subset \eta^- \backslash x}B^-(x,\zeta^- \cup \xi^-)d\lambda(\xi^-) 
 \end{align*}
 where the latter expression is well-defined due to $M_-(\eta^-) > 0$ for $\eta^- \neq \emptyset$.
 Using $(L^{\Delta,-}k)(\emptyset) = 0$, it is easily seen that $k$ satisfies \eqref{EQ:21} if and only if 
 \begin{align}\label{EQ:03}
  \widetilde{k}(\eta^-) - (S\widetilde{k})(\eta^-) = \rho \1_{\Gamma_1^-}(\eta^-)\sum \limits_{x \in \eta^-}\frac{B^-(x,\emptyset)}{D^-(x,\emptyset)}.
 \end{align}
 It is not difficult to see that $S$ is a bounded linear operator such that $\| S \|_{L(\K_{C_-})} < 1$.
 \qed
\end{proof}
\begin{remark}
 Equation \eqref{EQ:03} is an analogue of the Kirkwood-Salsburg equation.
\end{remark}
Next we prove that $\widehat{L}^-$ has a spectral gap. For this purpose we introduce the same decomposition as \eqref{EQ:22} for $\Lb_{C_-}$, i.e.
$\Lb_{C_-} = \Lb_{C_-}^{0} \oplus \Lb_{C_-}^{\geq 1}$ with projection operators
\begin{align*}
 P_0G(\eta^-) = 0^{|\eta^-|}G(\emptyset), \quad \quad P_{\geq 1}G(\eta^-) = (1 - 0^{|\eta^-|})G(\eta^-).
\end{align*}
\begin{proposition}\label{TH:00}
 Suppose that conditions (E1), (E2) and \eqref{EQ:05} are satisfied and let
 \begin{align}\label{GMCS:37}
  \omega_0 := \sup \left\{ \omega \in \left[0, \frac{\pi}{4}\right] \ \bigg | \ a_- < 1+ \cos(\omega) \right\}.
 \end{align}
 Then the following statements hold
 \begin{enumerate}
  \item[(a)] The point $0$ is an eigenvalue for $(\widehat{L}^-, D(\widehat{L}^-))$ with eigenspace $\Lb_{C_-}^0$.
  \item[(b)] Let $\lambda_0 := (2 - a_-)M_* > 0$ where $M_* := \inf_{\emptyset \neq \eta^- \in \Gamma_0^-}M_-(\eta)$, then
  \begin{align*}
   I_1 &:= \{ \lambda \in \C \ | \ \mathrm{Re}(\lambda) > - \lambda_0 \} \backslash \{0\},
   \\ I_2 &:= \left\{ \lambda \in \C \ \bigg| \ |\mathrm{arg}(\lambda)| < \frac{\pi}{2} + \omega_0 \right\} \backslash \{0\}
  \end{align*}
  both belong to the resolvent set $\rho(\widehat{L}^-)$ of $\widehat{L}^-$ on $\Lb_{C_-}$.
 \end{enumerate}
\end{proposition}
\begin{proof}
 Observe that $\widehat{L}^-P_0 = 0$ and hence $\widehat{L}^- P_{\geq 1} = \widehat{L}^-$. Thus we obtain the decomposition
 $\widehat{L}^- = L_{10}P_{\geq 1} + L_{11}P_{\geq 1}$, where $D(L_{11}) = D(\widehat{L}^-) \cap \Lb_{C_-}^{\geq 1}$ and
 \begin{align*}
  &L_{10}: D(L_{11}) \longrightarrow \Lb_{C_-}^{0}, \ \ L_{10}G = P_0 \widehat{L}^-G
  \\ &L_{11}: D(L_{11}) \longrightarrow \Lb_{C_-}^{\geq 1}, \ \ L_{11}G = P_{\geq 1}\widehat{L}^-G.
 \end{align*}
 Using the definition of $\widehat{L}^-$ we see that
 $(L_{10}G)(\eta^-) = 0^{|\eta^-|}\int_{\R^d}G(x)B^-(x,\emptyset)dx$ and $L_{11} = A_{11} + B_{11}$ with
 \begin{align*}
   (A_{11}G)(\eta^-) &= - M_-(\eta^-)G(\eta^-) 
 \\(B_{11}G)(\eta^-) &= - \sum \limits_{\xi^- \subsetneq \eta^-}G(\xi^-) \sum \limits_{x \in \xi^-}\sum \limits_{\zeta^- \subset \xi^- \backslash x}D^-(x,\eta^- \backslash \xi^- \cup \zeta^-)
 \\ &+ \sum \limits_{\xi^- \subset \eta^-}\int \limits_{\R^d}G(\xi^- \cup x)\sum \limits_{\zeta^- \subset \xi^-}B^-(x,\eta^- \backslash \xi^- \cup \zeta^-)d x.
 \end{align*}
 Let us first show that $(L_{11},D(L_{11}))$ is invertible on $\Lb_{C_-}^{\geq 1}$. Denote by $\Vert \cdot \Vert_{\Lb_{C_-}}$
 the norm on $\Lb_{C_-}^{\geq 1}$. Since $M_-(\eta^-) \geq M_*$ for all $|\eta^-| \geq 1$,
 we obtain for any $\lambda = u + i w$, $u \geq 0$, $w \in \R$
 \[
  \left| \frac{G}{\lambda + M_-(\eta^-)} \right| \leq \frac{|G|}{\sqrt{ (u + M_*)^2 + w^2 }} \leq |G| \min\left( \frac{1}{|\lambda|}, \frac{1}{\sqrt{M_*^2 + w^2}}\right).
 \]
 This implies $\lambda \in \rho(A_{11})$ and
 \begin{align}\label{GMCS:35}
  \Vert R(\lambda;A_{11})G \Vert_{\Lb_{C_-}} \leq \min\left( \frac{1}{|\lambda|}, \frac{1}{\sqrt{M_*^2 + w^2}}\right)\Vert G \Vert_{\Lb_{C_-}}.
 \end{align} 
 A simple computation shows that for any $G \in \Lb_{C_-}^{\geq 1}$
 \begin{align*}
  \Vert B_{11}G \Vert_{\Lb_{C_-}} 
  \leq \int \limits_{\Gamma_0^-} |B_{11}G(\eta^-)|C_-^{|\eta^-|}\dm \lambda(\eta^-)
  \leq (a_- - 1) \Vert A_{11} G \Vert_{\Lb_{C_-}}.
 \end{align*}
 Hence we see that $(1 - B_{11}R(\lambda;A_{11}))$ is invertible on $\Lb_{C_-}^{\geq 1}$ and using
 \begin{align}\label{GMCS:36}
  (\lambda - L_{11}) = (1 - B_{11}R(\lambda; A_{11}))(\lambda - A_{11}).
 \end{align}
 we obtain $\lambda \in \rho(L_{11})$ with
 \begin{align}\label{GMCS:34}
  R(\lambda; L_{11}) = R(\lambda;A_{11}) (1 - B_{11}R(\lambda; A_{11}))^{-1}.
 \end{align}
 In particular, we obtain for $\lambda = u + iw$, $u \geq 0$, $w \in \R$ by \eqref{GMCS:35} and \eqref{GMCS:34}
 \[
  \Vert R(\lambda;L_{11})G \Vert_{\Lb_{C_-}} \leq \frac{ \min\left( \frac{1}{|\lambda|}, \frac{1}{\sqrt{M_*^2 + w^2}}\right)}{2 - a_-}\Vert G \Vert_{\Lb_{C_-}}
 \]
 and for $\lambda = iw$, $w \in \R$
 \[
  \Vert R(iw, L_{11})G \Vert_{\Lb_{C_-}} \leq \frac{\sqrt{M_*^2 + w^2}^{-1}}{2 - a_-}\Vert G \Vert_{\Lb_{C_-}}.
 \]
 For $\lambda = u + iw$, $0 > u > - \lambda_0$ and $w \in \R$ write
 \[
  ( u +iw - L_{11})= (1 + u R(iw;L_{11}))(iw - L_{11}).
 \]
 Then, by $|u| < \lambda_0$ and $\frac{|u|}{\sqrt{M_*^2 + w^2}}\frac{1}{2-a_-} \leq \frac{|u|}{\lambda_0} < 1$ we obtain $\lambda \in \rho(L_{11})$ and
 \[
  \Vert R(\lambda; L_{11}) G \Vert_{\Lb_{C_-}} \leq \frac{\sqrt{M_*^2 + w^2}^{-1}}{2 - a_-} \left(1 - \frac{|u|}{\lambda_0}\right)^{-1}\Vert G \Vert_{\Lb_{C_-}}.
 \]
 Therefore, $I_1$ belongs to the resolvent set of $L_{11}$. 
 For $I_2$ let $\lambda = u + iw \in I_2$ and $u < 0$. Then, there exists $\omega \in (0,\omega_0)$ such that $|\mathrm{arg}(\lambda)| < \frac{\pi}{2} + \omega$ and hence
 \[
  |w| = |\cot(\arg(\lambda)- \frac{\pi}{2})| |u| \geq \cot(\omega) |u|.
 \]
 This implies for $\eta^- \neq \emptyset$
 \[
  |\lambda + M_-(\eta^-)|^2 = (u + M_-(\eta^-))^2 + w^2 \geq  (u + M_-(\eta^-))^2 + \cot(\omega)^2 u^2.
 \]
 The right-hand side is minimal for the choice $u = - \frac{M_-(\eta^-)}{1 + \cot(\omega)^2}$ which yields
 \begin{align*}
  |\lambda + M_-(\eta^-)|^2 &\geq M_-(\eta^-)^2\left( \left( \frac{\cot(\omega)^2}{ 1 + \cot(\omega)^2}\right)^2 + \frac{\cot(\omega)^2}{(1+\cot(\omega)^2)^2}\right)
  \\ &= M_-(\eta^-)^2 \frac{\cot(\omega)^2}{1+\cot(\omega)^2} = M_-(\eta^-)^2 \cos(\omega)^2.
 \end{align*}
 Then
 \[
  \Vert B_{11}R(\lambda;A_{11})G\Vert_{\Lb_{C_-}} \leq (a_- - 1) \Vert A_{11} R(\lambda;A_{11})G\Vert_{\Lb_{C_-}} \leq \frac{a_- - 1}{\cos(\omega)}\Vert G \Vert_{\Lb_{C_-}}.
 \]
 Finally by \eqref{GMCS:36} together with $a_- - 1 < \cos(\omega)$ (see \eqref{GMCS:37}) we obtain $I_2 \subset \rho(L_{11})$. 
 Moreover, for each $\lambda = u + iw$ such that $\frac{\pi}{2} < |\mathrm{arg}(\lambda)| < \frac{\pi}{2} + \omega$ and, for some $\omega \in (0,\omega_0)$,
 \begin{align*}
  \Vert R(\lambda;L_{11})G \Vert_{\Lb_{C_-}} &\leq \frac{\sqrt{(u^2 + M_*^2)^2 + w^2}^{-1}}{1 - \frac{a_- - 1}{\cos(\omega)}} \Vert G \Vert_{\Lb_{C_-}}
  \\ &\leq \frac{(1 - \frac{a_- - 1}{\cos(\omega)})^{-1}}{|w|} \Vert G \Vert_{\Lb_{C_-}}
  \leq \sqrt{2}\frac{(1 - \frac{a_- - 1}{\cos(\omega)})^{-1}}{|\lambda|} \Vert G \Vert_{\Lb_{C_-}},
 \end{align*}
 where we have used $|w| \geq \frac{|\lambda|}{\sqrt{2}}$.
 
 Let us prove (a). Take $\psi \in D(\widehat{L}^-)$ and consider the decomposition $\psi = P_0\psi + P_{\geq 1}\psi = \psi_0 + \psi_1$ with
 $\psi_0 \in \Lb_{C_-}^{0}$ and $\psi_1 \in D(L_{11})$. Then
 \[
  0 = \widehat{L}^-\psi = L_{01}\psi_1 + L_{11}\psi_1 \in \Lb_{C_-}^0 \oplus \Lb_{C_-}^{\geq 1}
 \]
 and hence $L_{11}\psi_1 = 0$. Since $0 \in \rho(L_{11})$ we obtain $\psi_1 = 0$.
 
 Let us prove (b). Let $\lambda \in I_1 \cup I_2$ and $H = H_0 + H_1 \in \Lb_{C_-}^0 \oplus \Lb_{C_-}^{\geq 1}$. 
 Then, we have to find $G \in D(\widehat{L}^-)$ such that $(\lambda - \widehat{L})G = H$.
 Using again the decomposition $\widehat{L}^- = L_{01}P_{\geq 1} + L_{11}P_{\geq 1}$, above equation is equivalent to the system of equations
 \begin{align*}
  \lambda G_0 - L_{01}G_1 &= H_0
  \\ (\lambda - L_{11})G_1 &= H_1.
 \end{align*}
 Since $\lambda \in I_1 \cup I_2 \subset \rho(L_{11})$ the second equation has a unique solution on $\Lb_{C_-}^{\geq 1}$ given by $G_1 = R(\lambda;L_{11})H_1$.
 Therefore, $G_0$ is given by
 \[
  G_0 = \frac{1}{\lambda}\left( H_0 + L_{01}R(\lambda; L_{11})H_1\right).
 \]
 \qed
\end{proof}
Define a projection operator $\widehat{P}: \Lb_{C_-} \longrightarrow \Lb_{C_-}^0$ by
\begin{align*}
 \widehat{P}G(\eta^-) = \int \limits_{\Gamma_0^-}G(\xi^-)k_{\mathrm{inv}}(\xi^-)d \lambda(\xi^-) 0^{|\eta^-|}
\end{align*}
Then $\langle \widehat{P}G, k \rangle = \langle G, \widehat{P}^*k\rangle$ where $\widehat{P}^*k(\eta^-) = k_{\mathrm{inv}}(\eta^-)k(\emptyset)$.
The next proposition completes the proof of Theorem \ref{TH:ERGODICITY}.
\begin{proposition}\label{PROP:02}
 There exists a unique invariant measure $\mu_{\mathrm{inv}} \in \mathcal{P}_{C_-}$ with correlation function $k_{\mathrm{inv}}$.
 Moreover, the following holds.
 \begin{enumerate}
  \item[(a)] $T^-(t)$ is uniformly ergodic with exponential rate and projection operator $\widehat{P}$.
  \item[(b)] $T^-(t)^*$ is uniformly ergodic with exponential rate and projection operator $\widehat{P}^*$.
 \end{enumerate}
\end{proposition}
\begin{proof}
Using $\widehat{L}^-P_0 = 0$ we obtain $T^-(t)P_0 = P_0$ and hence
\begin{align}\label{GMCS:67}
 T^-(t) = P_0 + P_0T^-(t)P_{\geq 1} + P_{\geq 1}T^-(t)P_{\geq 1}, \ \ t \geq 0.
\end{align}
The $\Lb_{C_-}^{\geq 1}$ part of $T^-(t)$ is given by $P_{\geq 1}T^-(t)P_{\geq 1}$ and has the generator $L_{11}$.
The proof of previous proposition shows that for any $\e > 0$ there exists $\omega = \omega(\e) \in (0, \frac{\pi}{2})$ such that
\[
 \Sigma(\e) := \left\{ \lambda \in \C \ \bigg| \ |\mathrm{arg}(\lambda + \lambda_0 - \e)| \leq \frac{\pi}{2} + \omega \right\} \subset I_1 \cup I_2 \cup \{0\}
\]
and there exists $M(\e) > 0$ such that
\[
 \Vert R(\lambda; L_{11})G\Vert_{\Lb_{C}^{\geq 1}} \leq \frac{M(\e)}{|\lambda|} \Vert G \Vert_{\Lb_{C_-}}
\]
for all $\lambda \in \Sigma(\e) \backslash \{0\}$. Moreover, $(L_{11}, D(L_{11}))$ is a sectorial operator of angle $\omega_0$ on $\Lb_{C_-}^{\geq 1}$.
Denote by $\widetilde{T}^-(t)$ the bounded analytic semigroup on $\Lb_{C_-}^{\geq 1}$ given by
\begin{align}\label{GMCS:53}
 \widetilde{T}^-(t) = \frac{1}{2\pi i} \int \limits_{\sigma} e^{\zeta t}R(\zeta;L_{11}) d \zeta, \ \ t > 0,
\end{align}
where the integral converges in the uniform operator topology, see \cite{PAZ83}. Here $\sigma$ denotes any piecewise smooth curve in
\[
 \left\{ \lambda \in \C \ \bigg| \ |\mathrm{arg}(\lambda)| < \frac{\pi}{2}+\omega_0 \right\} \backslash \{0\}
\]
running from $\infty e^{-i\theta}$ to $\infty e^{i \theta}$ for $\theta \in (\frac{\pi}{2}, \frac{\pi}{2}+ \omega_0)$.
Then $\widetilde{T}^-(t) = P_{\geq 1}T^-(t)P_{\geq 1}$.

The spectral properties stated above and \eqref{GMCS:53}
imply that for any $\e > 0$ there exists $C(\e) > 0$ such that for any $t \geq 0$ and $G \in \Lb_{C_-}^{\geq 1}$
\[
 \Vert \widetilde{T}^-(t)G\Vert_{\Lb_{C_-}} \leq C(\e)e^{-(\lambda_0 - \e)t}\Vert G \Vert_{\Lb_{C_-}}.
\]
Repeat, e.g., the arguments in \cite{KKM10}. 
By duality and \eqref{GMCS:67}, we see that the adjoint semigroup $(T^-(t)^*)_{t \geq 0}$ admits the decomposition
\begin{align*}
 T^-(t)^* = P_0 + P_{\geq 1}T^-(t)^*P_0 + \widetilde{T}^-(t)^*, \ \ t \geq 0,
\end{align*}
where $\widetilde{T}^-(t)^* \in L(\K_{C_-}^{\geq 1})$ is the adjoint semigroup to $(\widetilde{T}^-(t))_{t \geq 0}$. Hence
\[
 \Vert T^-(t)^*k \Vert_{\K_{C_-}} \leq C(\e)e^{-(\lambda_0 - \e)t}\Vert k \Vert_{\K_{C_-}}, \ \ k \in \K_{C_-}^{\geq 1}.
\]
Let $k \in \K_{C_-}$, then $k - \widehat{P}^*k \in \K_{C_-}^{\geq 1}$.
Using $T^-(t)^*\widehat{P}^* = \widehat{P}^*T^-(t)^* = \widehat{P}^*$ we obtain
\begin{align*}
 \Vert T^-(t)^*k - \widehat{P}^*k \Vert_{\K_{C_-}} = \Vert T^-(t)^*(k - \widehat{P}^*k)\Vert_{\K_{C_-}}
 \leq C(\e)e^{-(\lambda_0 - \e)t}\Vert k - \widehat{P}^*k \Vert_{\K_{C_-}}.
\end{align*}
This shows that $T^-(t)^*$ is uniformly ergodic with exponential rate. Duality implies that $T^-(t)$ is uniformly ergodic with exponential rate.
Let $\mu_0 \in \mathcal{P}_{C_-}$, $\mu_t \in \mathcal{P}_{C_-}$ the associated evolution of states and 
$k_{\mu_t} \in \K_{C_-}$ its correlation function for $t \geq 0$. By 
\[
 \langle G, k_{\mu_t}\rangle = \langle G, T^-(t)^*k_{\mu_0} \rangle = \langle T^-(t)G, k_{\mu_0}\rangle \geq 0
\]
for any $G \in B_{bs}(\Gamma_0^-)$ with $KG \geq 0$ and the ergodicity for $T^-(t)$ we see that $k_{\mathrm{inv}}$ is positive definite.
Thus, there exists a unique measure $\mu_{\mathrm{inv}} \in \mathcal{P}_{C_-}$ having $k_{\mathrm{inv}}$ as its correlation function. 
\qed
\end{proof}

\section{Proof: Stochastic averaging principle}
In this section we suppose that (E1) -- (E3), (S1) -- (S3), (AV1) -- (AV3) and \eqref{EQ:05} are satisfied.
Introduce the Banach space $\Lb_{C}$ of equivalence classes of integrable functions on $\Gamma_0^2$ equipped with the norm
\[
 \| G \|_{\Lb_{C}} = \int \limits_{\Gamma_0^2}|G(\eta)|C_+^{|\eta^+|}C_-^{|\eta^-|}d\lambda(\eta).
\]
Then $\Lb_{C} \cong \Lb_{C_+} \widehat{\otimes}_{\pi} \Lb_{C_-}$ where $\widehat{\otimes}_{\pi}$ denotes the projective tensor product of Banach spaces.
Given bounded linear operators $A_1$ on $\Lb_{C_+}$ and $A_2$ on $\Lb_{C_-}$,
the product $A_1 \otimes A_2$ on $\Lb_{C}$ is defined as the unique linear extension of the operator
\[
 (A_1 \otimes A_2)G(\eta) = A_1G_1(\eta^+) A_2G_2(\eta^-), \ \ G \in \mathcal{X}.
\]
This definition satisfies $\Vert (A_1 \otimes A_2) G \Vert_{\Lb_{C}} = \Vert A_1 G_1 \Vert_{\Lb_{C_+}}\Vert A_2 G_2 \Vert_{\Lb_{C_-}}$.
Since $\Lb_{C} \cong \Lb_{C_+} \widehat{\otimes}_{\pi} \Lb_{C_-}$, one can show that such extension exists (see \cite{R02}).
For $A_2$ being the identity operator we use the notation $A_1 \otimes \1$ and for $A_1$ being the identity we use the notation $\1 \otimes A_2$ respectively.

\subsection*{Step 1. Construction of isolated, ergodic environment}
In this step we study, in contrast to Theorem \ref{TH:FPE} and Theorem \ref{TH:ERGODICITY}, the environment process on $\Gamma^2$.
We study the extension of $\widehat{L}^-$ (introduced in previous section) onto $\Lb_C$.
Namely, let
\begin{align}\label{EQ:15}
 (\widehat{L}^- \otimes \1 )G(\eta) := &- \sum \limits_{\xi^- \subset \eta^-}G(\eta^+,\xi^-) \sum \limits_{x \in \xi^-}\sum \limits_{\zeta^- \subset \xi^- \backslash x}D^-(x,\eta^- \backslash \xi^- \cup \zeta^-)
 \\ \notag &+ \sum \limits_{\xi^- \subset \eta^-}\int \limits_{\R^d}G(\eta^+,\xi^- \cup x)\sum \limits_{\zeta^- \subset \xi^-}B^-(x,\eta^- \backslash \xi^- \cup \zeta^-)d x.
\end{align}
be defined on the domain $D(\1 \otimes \widehat{L}^-) = \{ G \in \Lb_{C} \ | \ M_- \cdot G \in \Lb_{C} \}$.
Since $\mathbb{K}$ is defined component-wise, it is easily seen that $\mathbbm{K}(\1 \otimes \widehat{L}^-) = L^-\mathbbm{K}$ holds on $B_{bs}(\Gamma_0^2)$,
where $L^-$ is extended onto $\mathcal{FP}(\Gamma^2)$ in the obvious way.
Let $\mathcal{X} := \{ G^1 \otimes G^2 \ | \ G^1 \in \Lb_{C_+}, \ G^2 \in \Lb_{C_-} \} \subset \Lb_{C}$
where $(G^1 \otimes G^2)(\eta) := G^1(\eta^+)G^2(\eta^-)$. Then, $\mathrm{lin}(\mathcal{X}) \subset \Lb_{C}$ is dense,
where $\mathrm{lin}$ denotes the linear span of a given subset of $\Lb_{C}$.
\begin{lemma}\label{MCSLLEMMA:00}
 The following assertions hold
 \begin{enumerate}
  \item[(a)] $\1 \otimes T^-(t)$ is an analytic semigroup of contractions on $\Lb_C$ with generator $\1 \otimes \widehat{L}^-$
  and core $B_{bs}(\Gamma_0^2)$.
  \item[(b)] Let $\mathcal{D} := \{ G^1 \otimes G^2 \in \mathcal{X} \ | \ G^2 \in D(\widehat{L}^-) \}$. 
  Then $\mathrm{lin}(\mathcal{D})$ is a core for the generator $(\1 \otimes \widehat{L}^-, D(\1 \otimes \widehat{L}^-))$.
 \end{enumerate}
\end{lemma}
\begin{proof}
 Following the same arguments as given in Proposition \ref{PROP:03} we easily deduce that $\1 \otimes \widehat{L}^-$ is the 
 generator of an analytic semigroup $U^-(t)$ of contractions on $\Lb_C$ and that $B_{bs}(\Gamma_0^2)$ is a core.
 It follows from the definition of $\1 \otimes \widehat{L}^-$ that 
 \begin{align*}
  (\1 \otimes \widehat{L}^-)(G^1 \otimes G^2) = G^1 \otimes (\widehat{L}^-G^2), \ \ G^1 \otimes G^2 \in \mathcal{D}.
 \end{align*}
 Hence $G_t := G^1 \otimes T^-(t)G^2$ is a solution to the Cauchy problem
 \[
  \frac{d}{d t}G_t = (\1 \otimes \widehat{L}^-)G_t, \ \ G_t|_{t=0} = G^1 \otimes G^2 \in \mathcal{D}
 \]
 on $\Lb_{C}$. Since for $G^1 \otimes G^2 \in \mathcal{D} \subset D(\1 \otimes \widehat{L}^-)$ this Cauchy problem has the 
 unique solution on $\Lb_{C}$ given by $U^-(t)(G^1 \otimes G^1)$, it follows that $G^1 \otimes T^-(t)G^2 = U^-(t)(G^1 \otimes G^2)$, 
 which gives $U^-(t) = \1 \otimes T^-(t)$.
 From this we easily deduce that $\mathcal{D}$, and hence $\mathrm{lin}(\mathcal{D})$, is invariant for $\1 \otimes T^-(t)$.
 Thus it is a core for the generator $\1 \otimes \widehat{L}^-$.
 \qed
\end{proof}
The following is the main estimate for the first step.
\begin{proposition}
 There exist constants $C, \lambda > 0$ such that for any $G \in \Lb_{C}$
 \begin{align}\label{MCSL:20}
  \Vert (\1 \otimes T^-(t))G - \widehat{P}G \Vert_{\Lb_{C}} \leq C e^{-\lambda t}\Vert G \Vert_{\Lb_{C}}, \ \ t \geq 0,
 \end{align}
 where $k_{inv}$ is the correlation function for $\mu_{\mathrm{inv}}$ and
 \begin{align}\label{MCSL:25}
  (\1 \otimes \widehat{P})G(\eta) = \int \limits_{\Gamma_0^-}G(\eta^+, \xi^-)k_{\mathrm{inv}}(\xi^-)d \lambda(\xi^-)0^{|\eta^-|}.
 \end{align}
\end{proposition}
\begin{proof}
First observe that due to $\Lb_{C} \cong \Lb_{C_+} \widehat{\otimes}_{\pi} \Lb_{C_-}$ and by \cite{R02} we have
\[
 \Lb_{C} \cong \left \{ G = \sum \limits_{n=1}^{\infty}G_n^1 \otimes G_n^2 \ \bigg| \ (G_n^1)_{n \in \N} \subset \Lb_{C_+}, \ (G_n^2)_{n \in \N} \subset \Lb_{C_-}, \ \sum \limits_{n=1}^{\infty}\Vert G_n^1\Vert_{\Lb_{C_+}} \Vert G_n^2\Vert_{\Lb_{C_-}} < \infty \right\}.
\]
Let $G = \sum_{n=1}^{\infty}G_n^1 \otimes G_n^2 \in \Lb_{C}$ and set 
$G_N := \sum_{n=1}^{N}G_n^1 \otimes G_n^2 \in \mathrm{lin}(\mathcal{X})$. 
By \eqref{MCSL:25} we get $(\1 \otimes \widehat{P})G_N = \sum_{n=1}^{N}G_n^1 \cdot \widehat{P} G_n^2$ and similarly
$(\1 \otimes T^-(t))G_N = \sum_{n=1}^{N}G_n^1\cdot T^-(t)G_n^2$. Thus we obtain
\begin{align*}
 \Vert (\1 \otimes T^-(t))G_N - (\1 \otimes \widehat{P})G_N \Vert_{\Lb_{C}}
 &\leq \sum \limits_{n=1}^{N}\Vert G_n^{1}\Vert_{\Lb_{C_+}} \Vert T^-(t)G_n^2 - \widehat{P}G_n^2 \Vert_{\Lb_{C_-}}
 \\ &\leq 2 C e^{-\lambda t} \sum \limits_{n=1}^{N}\Vert G_n^1\Vert_{\Lb_{C_+}} \Vert G_n^{2}\Vert_{\Lb_{C_-}},
\end{align*}
where we have used Proposition \ref{PROP:02}.(a).
Taking the limit $N \to \infty$ yields by $G_N \longrightarrow G$ in $\Lb_{C}$
\begin{align}\label{EQ:42}
 \Vert (\1 \otimes T^-(t))G - \widehat{P}G \Vert_{\Lb_{C}} \leq 2 C e^{-\lambda t} \sum \limits_{n=1}^{\infty}\Vert G_n^1\Vert_{\Lb_{C_+}}\Vert G_n^2\Vert_{\Lb_{C_-}}.
\end{align}
Since (see \cite{R02})
\[
 \Vert G \Vert_{\Lb_{C}} = \inf \left\{ \sum \limits_{n=1}^{\infty}\Vert G_n^1\Vert_{\Lb_{C_+}} \Vert G_n^2\Vert_{\Lb_{C_-}} \ \bigg | \ G = \sum \limits_{n=1}^{\infty}G_n^1 \otimes G_n^2 \in \Lb_{C} \right\}
\]
we find a sequence $G_k = \sum_{n=1}^{\infty}G_{n,k}^1 \otimes G_{n,k}^2$ with 
$\sum_{n=1}^{\infty}\Vert G_{n,k}^1\Vert_{\Lb_{C_+}} \Vert G_{n,k}^2\Vert_{\Lb_{C_-}} \longrightarrow \Vert G \Vert_{\Lb_{C}}$, as $k \to \infty$.
Applying \eqref{EQ:42} to $G_k$ gives
\[
 \Vert (\1 \otimes T^-(t))G_k - (\1 \otimes \widehat{P})G_k \Vert_{\Lb_{C}} \leq 2 C e^{-\lambda_0 t} \sum \limits_{n=1}^{\infty}\Vert G_{n,k}^1\Vert_{\Lb_{C_+}}\Vert G_{n,k}^2\Vert_{\Lb_{C_-}}.
\]
Taking the limit $k \to \infty$ yields \eqref{MCSL:20}.
\qed
\end{proof}

\subsection*{Step 2. Construction of scaled dynamics}
Consider the scaled operator $L_{\e} = L^+ + \frac{1}{\e}L^-$ and let 
$\widehat{L}_{\e} = \widehat{L}^+ + \frac{1}{\e}\widehat{L}^-$ where $\widehat{L}^-$ is given by \eqref{EQ:15} and 
\begin{align*}
 (\widehat{L}^+G)(\eta) = - \sum \limits_{\xi \subset \eta}G(\xi) \sum \limits_{x \in \xi^+}\sum \limits_{\bfrac{\zeta^+ \subset \xi^+ \backslash x}{\zeta^- \subset \xi^-}}D^+(x,\eta \backslash \xi \cup \zeta)
 + \sum \limits_{\xi \subset \eta}\int \limits_{\R^d}G(\xi^+ \cup x, \xi^-)\sum \limits_{\zeta \subset \xi}B^+(x,\eta \backslash \xi \cup \zeta)d x.
\end{align*}
is defined on $D(\widehat{L}^+) = \{ G \in \Lb_{C} \ | \ M_+ \cdot G \in \Lb_{C}\}$ with
$M_+(\eta) = \sum_{x \in \eta^+}d^+(x,\eta^+ \backslash x, \eta^-)$.
\begin{lemma}
 The following assertions hold
 \begin{enumerate}
  \item[(a)] The operator $(\widehat{L}^+, D(\widehat{L}^+))$ is the generator of an analytic semigroup $(T^+(t))_{t \geq 0}$ of contractions on $\Lb_{C}$.
  Moreover, the assertions from Theorem \ref{TH:FPESYSTEM} are satisfied.
  \item[(b)] $D(\widehat{L}^+) \cap D(\1 \otimes \widehat{L}^-)$ is a core for the operator $(\1 \otimes \widehat{L}^-, D(\1 \otimes \widehat{L}^-))$.
 \end{enumerate}
\end{lemma}
\begin{proof}
 (a) This follows by the same arguments as given in section four (see also \cite{FK16b}).
 \\ (b) Let $G \in D(\1 \otimes \widehat{L}^-)$ and set $G_{\lambda} = \frac{\lambda}{\lambda + M_+}G$, $\lambda > 0$. 
 Then $G_{\lambda} \in D(\widehat{L}^+) \cap D(\1 \otimes \widehat{L}^-)$ and obviously $G_{\lambda} \longrightarrow G$ in $\Lb_{C}$ as $\lambda \to \infty$.
 Moreover, we have
 \begin{align*}
  \Vert (\1 \otimes \widehat{L}^-)G_{\lambda} - (\1 \otimes \widehat{L}^-)G \Vert_{\Lb_{C}}
  &\leq a_- \Vert M_- \cdot( G_{\lambda} - G) \Vert_{\Lb_{C}}
  \\ &= a_- \int \limits_{\Gamma_0^2} \frac{M_-(\eta^-) M_+(\eta)}{\lambda + M_+(\eta)} |G(\eta)| C_+^{|\eta^+|}C_-^{|\eta^-|}d \lambda(\eta)
 \end{align*}
 which tends by dominated convergence to zero as $\lambda \to \infty$.
 \qed
\end{proof}
Finally we may show the following.
\begin{proposition}\label{TH:02}
 For every $\e > 0$ the operator $(\widehat{L}^+ + \frac{1}{\e}\widehat{L}^-, D(\widehat{L}^+) \cap D(\widehat{L}^-))$ 
 is the generator of an analytic semigroup of contractions on $\Lb_{C}$.
 Moreover, the assertions from Theorem \ref{TH:MAIN}.(a) are satisfied.
\end{proposition}
\begin{proof}
 For $\e > 0$ and $\eta \in \Gamma_0^2$ we get
 \begin{align*}
  c_+(\eta) + \frac{1}{\e}c_-(\eta^-) &\leq a_{+}M_{+}(\eta) + \frac{a_-}{\e}M_{-}(\eta^-) 
  \leq \max\ \{ a_-, a_{+}\} \left(M_{+}(\eta) + \frac{1}{\e}M_{-}(\eta^-)\right).
 \end{align*}
 The assertion can be now deduced by the same arguments as given in section four (see also \cite{FK16b}).
 \qed
\end{proof}

\subsection*{Step 3. Construction of averaged dynamics}
Define a linear operator by
\begin{align*}
 \widehat{\overline{L}}G(\eta) = &- \sum \limits_{\xi^+ \subset \eta^+}G(\xi^+) \sum \limits_{x \in \xi^+}\sum \limits_{\zeta^+ \subset \xi^+ \backslash x}\overline{D}(x,\eta^+ \backslash \xi^+ \cup \zeta^+)
 \\ &\ \ \ + \sum \limits_{\xi^+ \subset \eta^+}\int \limits_{\R^d}G(\xi^+ \cup x)\sum \limits_{\zeta^+ \subset \xi^+}\overline{B}(x,\eta^+ \backslash \xi^+ \cup \zeta^+)d x.
\end{align*}
on the domain $D(\widehat{\overline{L}}) = \{ G \in \Lb_{C_+} \ | \ G \cdot \overline{M} \in \Lb_{C_+} \}$
with $\overline{M}(\eta^+) = \sum_{x \in \eta^+}\overline{d}(x,\eta^+ \backslash x)$.
Following line to line the arguments in section four we readily deduce Theorem \ref{TH:MAIN}.(b) and, in particular, the next proposition.
\begin{proposition}
 The operator $(\widehat{\overline{L}}, D(\widehat{\overline{L}}))$ is the generator of an analytic semigroup $\overline{T}(t)$ on $\Lb_{C_+}$.
 Moreover, $\overline{T}(t)$ is a contraction operator and $B_{bs}(\Gamma_0^+)$ is a core for the generator.
\end{proposition}

\subsection*{Step 4. Stochastic averaging principle}
Using the definition of correlation functions and previous steps, we immediately see that the main result, Theorem \ref{TH:MAIN}.(c),
follows from the next proposition.
\begin{proposition}\label{MCSLTH:06}
 Let $T^{\e}(t)$ be the semigroup generated by $(\widehat{L}^+ + \frac{1}{\e}\widehat{L}^-, D(\widehat{L}^+)\cap D(\widehat{L}^-))$
 and let $\overline{T}(t)$ be the semigroup generated by $(\widehat{\overline{L}}, D(\widehat{\overline{L}}))$.
 Then for any $T > 0$ and $G \in \Lb_{C_+}$ 
 \begin{align}\label{MCSL:39}
  \lim \limits_{\e \to 0} \sup \limits_{t \in [0,T]} \| T^{\e}(t)(G \otimes 0^-) - (\overline{T}(t)G)\otimes 0^- \|_{\Lb_C} = 0,
 \end{align}
 where $(G \otimes 0^-)(\eta) = G(\eta^+)0^{|\eta^-|}$ and for any $k \in \K_{C}$
 \begin{align}\label{MCSL:40}
  \lim \limits_{\e \to 0}\sup \limits_{t \in [0,T]} \left|\int \limits_{\Gamma_0^+}G(\eta^+)(T^{\e}(t)^*k)(\eta^+, \emptyset) d \lambda(\eta^+) - \int \limits_{\Gamma_0^+}G(\eta^+)\left( \overline{T}(t)^*k(\cdot, \emptyset)\right)(\eta^+)d \lambda(\eta^+)\right| = 0.
 \end{align}
\end{proposition}
\begin{proof}
 Suppose that the following properties are satisfied
 \begin{enumerate}
  \item[(i)] $(\widehat{L}^+, D(\widehat{L}^+))$ is the generator of a strongly continuous contraction semigroup.
  \item[(ii)] $(\1 \otimes \widehat{L}^{-}, D(\1 \otimes \widehat{L}^{-}))$ is the generator of a strongly continuous contraction semigroup 
  $\1 \otimes T^-(t)$ on $\Lb_C$. Moreover $\1 \otimes T^-(t)$ is ergodic with projection operator $\1 \otimes \widehat{P}$
  and $D(\widehat{L}^+) \cap D(\1 \otimes \widehat{L}^-)$ is a core for the generator.
  \item[(iii)] $(\widehat{\overline{L}}, D(\widehat{\overline{L}}))$ is the generator of a strongly continuous contraction semigroups $\overline{T}(t)$ on $\Lb_{C_+}$.
  \item[(iv)] The averaged operator $(\1 \otimes \widehat{P})\widehat{L}^+$ equipped with the domain $\mathrm{Ran}(\1 \otimes \widehat{P}) \cap D(\widehat{L}^+)$
  is closable and its closure $(\mathcal{PL}^+, D(\mathcal{PL}^+))$ satisfies
  \[
   (\mathcal{PL}^+G)(\eta) = 0^{|\eta^-|}(\widehat{\overline{L}}G)(\eta^+), \ \ G \otimes 0^- \in D(\mathcal{PL}^+)
  \]
  where $D(\mathcal{PL}^+) = \{ G \otimes 0^- \ | \ G \in D(\widehat{\overline{L}}) \}$.
 \end{enumerate}
 Then we may apply \cite[Theorem 2.1]{KURTZ73} and from that readily deduce \eqref{MCSL:39}.
 By duality this yields
 \begin{align*}
  \langle G \otimes 0^-, T^{\e}(t)^*k \rangle = \langle T^{\e}(t)(G \otimes 0^-), k\rangle
  \longrightarrow \langle (\overline{T}(t)G) \otimes 0^- , k\rangle
 \end{align*}
 uniformly on compacts in $t \geq 0$. Convergence \eqref{MCSL:40} then follows from
 \begin{align*}
  \int \limits_{\Gamma_0^+}\overline{T}(t)G(\eta^+) k(\eta^+,\emptyset)d \lambda(\eta^+) 
  = \int \limits_{\Gamma_0^+}G(\eta^+)\left( \overline{T}(t)^*k(\cdot, \emptyset)\right)(\eta^+)d \lambda(\eta^+).
 \end{align*}
 Properties (i) -- (iii) have been checked in step 1 -- step 3. It remains to prove (iv).
 First observe that $\mathrm{Ran}(\1 \otimes \widehat{P}) = \Lb_{C_+}\otimes \Lb_{C_-}^0$ and hence by definition of $D(\widehat{L}^+)$
 \begin{align*}
  \mathrm{Ran}(\1 \otimes \widehat{P}) \cap D(\widehat{L}^+) 
  = \{ G = 0^{|\eta^-|}G_1 \in \Lb_{C_+}\otimes \Lb_{C_-}^0 \ |\ M_{+}(\cdot, \emptyset) \cdot G_1 \in \Lb_{C_+} \}.
 \end{align*}
For such $G(\eta) = G_1(\eta^+)0^{|\eta^-|} \in \mathrm{Ran}(\1 \otimes \widehat{P}) \cap D(\widehat{L}^+) $ we obtain
\begin{align*}
 (\widehat{L}^+G)(\eta) &= - \sum \limits_{\xi^+ \subset \eta^+}G_1(\xi) \sum \limits_{x \in \xi^+}\sum \limits_{\zeta^+ \subset \xi^+ \backslash x}D^+(x,\eta^+ \backslash \xi^+ \cup \zeta^+, \eta^-)
 \\ &\ \ \ + \sum \limits_{\xi^+ \subset \eta^+}\int \limits_{\R^d}G_1(\xi^+ \cup x)\sum \limits_{\zeta^+ \subset \xi^+}B^+(x,\eta^+ \backslash \xi^+ \cup \zeta^+, \eta^-)d x.
\end{align*}
and hence applying $\1 \otimes \widehat{P}$ yields by the definitions \eqref{MCSL:23} and \eqref{MCSL:24}
\begin{align*}
 ((\1 \otimes \widehat{P})\widehat{L}^+G)(\eta) &= - 0^{|\eta^-|}\sum \limits_{\xi^+ \subset \eta^+}G_1(\xi^+) \sum \limits_{x \in \xi^+}\sum \limits_{\zeta^+ \subset \xi^+ \backslash x}\overline{D}(x,\eta^+ \backslash \xi^+ \cup \zeta^+)
 \\ &\ \ \ + 0^{|\eta^-|}\sum \limits_{\xi^+ \subset \eta^+}\int \limits_{\R^d}G_1(\xi^+ \cup x)\sum \limits_{\zeta^+ \subset \xi^+}\overline{B}(x,\eta^+ \backslash \xi^+ \cup \zeta^+)d x
 \\ &= (\widehat{\overline{L}}G_1)(\eta^+)0^{|\eta^-|}.
\end{align*}
Similar arguments to section four (see also \cite{FK16b}) together with (AV1), (AV2) imply that 
the right-hand side equipped with the domain $\{ G \otimes 0^- \in \Lb_{C_+} \otimes \Lb_{C_-}^0 \ | \ G \in D(\widehat{\overline{L}}) \}$
is the generator of the analytic semigroup $\overline{U}(t)$ of contractions on $\Lb_{C_+} \otimes \Lb_{C_-}^0$ given by
\[
 (\overline{U}(t)(G \otimes 0^-))(\eta) := \overline{T}(t)G(\eta^+) 0^{|\eta^-|}.
\]
The generator has clearly 
\[
 \mathcal{D} := \{ G \otimes 0^- \in \Lb_{C_+} \otimes \Lb_{C_-}^0 \ | \ G \in B_{bs}(\Gamma_0^+) \}  
\]
as a core. Since $\mathcal{D} \subset \mathrm{Ran}(\1 \otimes \widehat{P}) \cap D(\widehat{L}^+)$ the assertion follows.
 \qed
\end{proof}

\subsection*{Acknowledgements}
Financial support through CRC701, project A5, at Bielefeld University is gratefully acknowledged.
The authors would like to thank the anonymous referee for his remarks which lead to a significant improvement of this work.

\newpage

\begin{footnotesize}

\bibliographystyle{alpha}
\bibliography{Bibliography}

\newcommand{\etalchar}[1]{$^{#1}$}
\def\cprime{$'$}
\begin{thebibliography}{FKKO15}

\bibitem[AR91]{AR91}
W.~Arendt and A.~Rhandi.
\newblock Perturbation of positive semigroups.
\newblock {\em Arch. Math. (Basel)}, 56(2):107--119, 1991.

\bibitem[BCF{\etalchar{+}}14]{BCFKKO14}
B.~Bolker, S.~Cornell, D.~Finkelshtein, Y.~Kondratiev, O.~Kutoviy, and
  O.~Ovaskainen.
\newblock A general mathematical framework for the analysis of spatio-temporal
  point processes.
\newblock {\em Theoretical Ecology}, 7(1):101--113, 2014.

\bibitem[EK86]{ET86}
S.~Ethier and T.~Kurtz.
\newblock {\em Markov processes}.
\newblock Wiley Series in Probability and Mathematical Statistics: Probability
  and Mathematical Statistics. John Wiley \& Sons, Inc., New York, 1986.
\newblock Characterization and convergence.

\bibitem[EW03]{EW03}
A.~Eibeck and W.~Wagner.
\newblock Stochastic interacting particle systems and nonlinear kinetic
  equations.
\newblock {\em Ann. Appl. Probab.}, 13(3):845--889, 2003.

\bibitem[FFH{\etalchar{+}}15]{FFHKKK15}
D.~Finkelshtein, M.~Friesen, H.~Hatzikirou, Y.~Kondratiev, T.~Kr\"uger, and
  O.~Kutoviy.
\newblock Stochastic models of tumour development and related mesoscopic
  equations.
\newblock {\em Inter. Stud. Comp. Sys.}, 7:5--85, 2015.

\bibitem[Fin11a]{F11}
D.~Finkelshtein.
\newblock Functional evolutions for homogeneous stationary death-immigration
  spatial dynamics.
\newblock {\em Methods Funct. Anal. Topology}, 17(4):300--318, 2011.

\bibitem[Fin11b]{F09}
D.~Finkelshtein.
\newblock Measures on two-component configuration spaces.
\newblock {\em Methods Funct. Anal. Topology}, 17(4):300--318, 2011.

\bibitem[FK16a]{FK16}
M.~Friesen and Y.~Kondratiev.
\newblock Weak-coupling limit for ergodic environments.
\newblock {\em Preprint: bibos.math.uni-bielefeld.de/preprints/16-08-509.pdf},
  2016.

\bibitem[FK16b]{FK16b}
Martin Friesen and Oleksandr Kutoviy.
\newblock Evolution of states and mesoscopic scaling for two-component
  birth-and-death dynamics in continuum.
\newblock {\em Methods Funct. Anal. Topology}, 22(4):346--374, 2016.

\bibitem[FKK12]{FKK12}
D.~Finkelshtein, Y.~Kondratiev, and O.~Kutoviy.
\newblock Semigroup approach to birth-and-death stochastic dynamics in
  continuum.
\newblock {\em J. Funct. Anal.}, 262(3):1274--1308, 2012.

\bibitem[FKK15]{FKK15}
D.~Finkelshtein, Y.~Kondratiev, and O.~Kutoviy.
\newblock Statistical dynamics of continuous systems: perturbative and
  approximative approaches.
\newblock {\em Arab. J. Math. (Springer)}, 4(4):255--300, 2015.

\bibitem[FKKO15]{FKK15WRMODEL}
D.~Finkelshtein, Y.~Kondratiev, O.~Kutoviy, and M.~J. Oliveira.
\newblock Dynamical {W}idom-{R}owlinson model and its mesoscopic limit.
\newblock {\em J. Stat. Phys.}, 158(1):57--86, 2015.

\bibitem[FKKZ12]{FKKZ12}
D.~Finkelshtein, Y.~Kondratiev, O.~Kutoviy, and E.~Zhizhina.
\newblock An approximative approach for construction of the {G}lauber dynamics
  in continuum.
\newblock {\em Math. Nachr.}, 285(2-3):223--235, 2012.

\bibitem[FKKZ14]{FKKZ14}
D.~Finkelshtein, Y.~Kondratiev, O.~Kutoviy, and E.~Zhizhina.
\newblock On an aggregation in birth-and-death stochastic dynamics.
\newblock {\em Nonlinearity}, 27(6):1105--1133, 2014.

\bibitem[FKO09]{FKO09}
D.~Finkelshtein, Y.~Kondratiev, and M.~J. Oliveira.
\newblock Markov evolutions and hierarchical equations in the continuum. {I}.
  {O}ne-component systems.
\newblock {\em J. Evol. Equ.}, 9(2):197--233, 2009.

\bibitem[FM04]{FM04}
N.~Fournier and S.~M{\'e}l{\'e}ard.
\newblock A microscopic probabilistic description of a locally regulated
  population and macroscopic approximations.
\newblock {\em Ann. Appl. Probab.}, 14(4):1880--1919, 2004.

\bibitem[Fri17]{F16WR}
M.~Friesen.
\newblock Non-equilibrium {D}ynamics for a {W}idom--{R}owlinson {T}ype {M}odel
  with {M}utations.
\newblock {\em J. Stat. Phys.}, 166(2):317--353, 2017.

\bibitem[GK06]{GKT06}
N.~L. Garcia and T.~G. Kurtz.
\newblock Spatial birth and death processes as solutions of stochastic
  equations.
\newblock {\em ALEA Lat. Am. J. Probab. Math. Stat.}, 1:281--303, 2006.

\bibitem[KK02]{KK02}
Y.~Kondratiev and T.~Kuna.
\newblock Harmonic analysis on configuration space. {I}. {G}eneral theory.
\newblock {\em Infin. Dimens. Anal. Quantum Probab. Relat. Top.},
  5(2):201--233, 2002.

\bibitem[KK06]{KK06}
Y.~Kondratiev and O.~Kutoviy.
\newblock On the metrical properties of the configuration space.
\newblock {\em Math. Nachr.}, 279(7):774--783, 2006.

\bibitem[KK16]{KK16}
Y.~Kondratiev and Y.~Kozitsky.
\newblock The evolution of states in a spatial population model.
\newblock {\em J. Dyn. Diff. Equat.}, 28(1):1--39, 2016.

\bibitem[KKM08]{KKM08}
Y.~Kondratiev, O.~Kutoviy, and R.~Minlos.
\newblock On non-equilibrium stochastic dynamics for interacting particle
  systems in continuum.
\newblock {\em J. Funct. Anal.}, 255(1):200--227, 2008.

\bibitem[KKM10]{KKM10}
Y.~Kondratiev, O.~Kutoviy, and R.~Minlos.
\newblock Ergodicity of non-equilibrium {G}lauber dynamics in continuum.
\newblock {\em J. Funct. Anal.}, 258(9):3097--3116, 2010.

\bibitem[Kol06]{KOL06}
V.~N. Kolokoltsov.
\newblock Kinetic equations for the pure jump models of {$k$}-nary interacting
  particle systems.
\newblock {\em Markov Process. Related Fields}, 12(1):95--138, 2006.

\bibitem[KS06]{KS06}
Y.~Kondratiev and A.~Skorokhod.
\newblock On contact processes in continuum.
\newblock {\em Infin. Dimens. Anal. Quantum Probab. Relat. Top.},
  9(2):187--198, 2006.

\bibitem[Kur73]{KURTZ73}
T.~G. Kurtz.
\newblock A limit theorem for perturbed operator semigroups with applications
  to random evolutions.
\newblock {\em J. Functional Analysis}, 12:55--67, 1973.

\bibitem[Kur92]{K92}
T.~G. Kurtz.
\newblock Averaging for martingale problems and stochastic approximation.
\newblock In {\em Applied stochastic analysis ({N}ew {B}runswick, {NJ}, 1991)},
  volume 177 of {\em Lecture Notes in Control and Inform. Sci.}, pages
  186--209. Springer, Berlin, 1992.

\bibitem[Len73]{L75a}
A.~Lenard.
\newblock Correlation functions and the uniqueness of the state in classical
  statistical mechanics.
\newblock {\em Comm. Math. Phys.}, 30:35--44, 1973.

\bibitem[Len75]{L75b}
A.~Lenard.
\newblock States of classical statistical mechanical systems of infinitely many
  particles. {II}. {C}haracterization of correlation measures.
\newblock {\em Arch. Rational Mech. Anal.}, 59(3):241--256, 1975.

\bibitem[Lot85]{L85}
Heinrich~P. Lotz.
\newblock Uniform convergence of operators on {$L^\infty$} and similar spaces.
\newblock {\em Math. Z.}, 190(2):207--220, 1985.

\bibitem[Paz83]{PAZ83}
A.~Pazy.
\newblock {\em Semigroups of linear operators and applications to partial
  differential equations}, volume~44 of {\em Applied mathematical sciences ;
  44}.
\newblock Springer, New York [u.a.], 2., corr. print edition, 1983.

\bibitem[Pin91]{PINSKY}
M.~A. Pinsky.
\newblock {\em Lectures on random evolution}.
\newblock World Scientific Publishing Co., Inc., River Edge, NJ, 1991.

\bibitem[Rya02]{R02}
R.~A. Ryan.
\newblock {\em Introduction to tensor products of {B}anach spaces}.
\newblock Springer Monographs in Mathematics. Springer-Verlag London, Ltd.,
  London, 2002.

\bibitem[SEW05]{SEW05}
D.~Steinsaltz, S.~N. Evans, and K.~W. Wachter.
\newblock A generalized model of mutation-selection balance with applications
  to aging.
\newblock {\em Adv. in Appl. Math.}, 35(1):16--33, 2005.

\bibitem[SHS02]{SHS}
A.~Skorokhod, F.~Hoppensteadt, and H.~Salehi.
\newblock {\em Random Perturbation Methods with Applications in Science and
  Engineering}.
\newblock Springer, 2002.

\bibitem[Spo88]{S88}
H.~Spohn.
\newblock Kinetic equations from {H}amiltonian dynamics: the {M}arkovian
  approximations.
\newblock In {\em Kinetic theory and gas dynamics}, volume 293 of {\em CISM
  Courses and Lectures}, pages 183--211. Springer, Vienna, 1988.

\bibitem[TV06]{TV06}
H.~R. Thieme and J.~Voigt.
\newblock Stochastic semigroups: their construction by perturbation and
  approximation.
\newblock In {\em Positivity {IV}---theory and applications}, pages 135--146.
  Tech. Univ. Dresden, Dresden, 2006.

\bibitem[WZ06]{WZ06}
L.~Wu and Y.~Zhang.
\newblock A new topological approach to the {$L^\infty$}-uniqueness of
  operators and the {$L^1$}-uniqueness of {F}okker-{P}lanck equations.
\newblock {\em J. Funct. Anal.}, 241(2):557--610, 2006.

\end{thebibliography}

\end{footnotesize}

\end{document}